\documentclass[11pt]{article}
\usepackage{textcomp}
\usepackage{setspace}
\usepackage{scalefnt}
\usepackage[vcentering,dvips]{geometry}
\usepackage{color}
\usepackage{natbib}
\usepackage{subfig}

\usepackage[normalem]{ulem}

\usepackage[usenames,dvipsnames]{xcolor}

\usepackage{graphicx}

\usepackage{xurl}

\newtheorem{theorem}{Theorem}

\newtheorem{corollary}{Corollary}

\newtheorem{lemma}{Lemma}
\newtheorem{proposition}{Proposition}

\newenvironment{proof}[1][Proof]{\noindent\textbf{#1.} }{\ \rule{0.5em}{0.5em}}

\newcommand{\E}{\mathbb{E}}

\newcommand{\dd}{\,d}

\usepackage{enumitem}
\usepackage{pdflscape}
\usepackage{amsmath}
\usepackage{amssymb}

\newcommand{\1}{\mathbf{1}}

\usepackage{cancel}

\begin{document}
\onehalfspacing

\title{Pure-Strategy Equilibrium in the \\Generalized First-Price Auction\thanks{We are grateful to Tatul Ayrapetyan, Sebasti\'{a}n Bauer, and Morteza Honarvar for excellent research assistance and to seminar participants at the 2022 NBER Market Design Working Group Meeting, the EC'25 Frontiers of Online Advertising Workshop, the 2026 International Workshop on Game Theory and Economic Applications, UC Irvine, Amazon, and Etsy for helpful comments and suggestions.}
}

\author{Michael Ostrovsky\thanks{ostrovsky@stanford.edu. Stanford University and NBER.}
\and Andrzej Skrzypacz\thanks{skrz@stanford.edu. Stanford University.}
}
\date{\today} 

\maketitle
\begin{abstract}
We revisit the classic result on the (non-)existence of pure-strategy Nash equilibria in the Generalized First-Price Auction for sponsored search advertising and show that the conclusion may be reversed when ads are ranked based on the product of stochastic quality scores and bid amounts, rather than solely on the bids or on the product of bids and deterministic quality scores. Moreover, the expected revenue in the pure strategy equilibrium of the Generalized First-Price Auction may substantially exceed that of the Generalized Second-Price Auction, although under some conditions the relation may also be reversed.
\end{abstract}

\newpage

\section{Introduction}
\label{sec:introduction}
The initial design of sponsored search auctions, implemented in 1997 by the company GoTo (later renamed Overture and subsequently acquired by Yahoo), used a very simple payment rule. Namely, each advertiser, following a click from a user, paid to GoTo the bid that it submitted to the system: e.g., if an advertiser submitted a bid of \$1, then that was the amount charged every time some user clicked on that advertiser's ad. 

This format, dubbed the Generalized First-Price Auction (GFP), was intuitive and easy to explain to the advertisers, but it suffered from a serious shortcoming. The ads on the results page were sorted purely based on the submitted bids (in the descending order), and so each advertiser, given a particular position, had an incentive to outbid the next highest competitor only by the smallest possible amount -- 1 cent. This incentive structure resulted in the non-existence of pure-strategy equilibria, and instead produced pronounced rapid ``cycling'' patterns, in which advertisers would continually outbid each other by 1 cent, then one advertiser would drop her bid by a large amount (when she decided that competing for the top position was no longer worth it, and settling for a lower one would be better) -- and then immediately her competitor would drop his bid as well (because he did not want to outbid her by more than 1 cent). This behavior led to stress on GoTo's ad servers, and also lowered ad auction efficiency and revenues \citep{EO2007}.\footnote{This cycling pattern is analogous to ``Edgeworth price cycles'' in the context of Bertrand price competition, in which competitors repeatedly undercut each other's price by a small amount, until one of them raises its price, restarting the cycle \citep{Edgeworth1925,Maskin1988,Eckert2004,Noel2007}. See \cite{ZF2011} for a detailed discussion of this connection.}

When Google came out with their own sponsored search auction design in 2002, they introduced two major changes. First, instead of ranking ads purely based on the corresponding bids, they ranked them based on the product of each ad's bid and its \textit{quality score}. Initially, quality score was equal to the estimated ``clickability'' of the ad, i.e., the probability that the ad would be clicked by a user if it was shown in the top position, although in later iterations it started including other factors (such as, e.g., the quality of the landing page). The second change was the introduction of a different payment rule. Instead of paying his own bid, each advertiser would pay the smallest amount to outbid his nearest competitor (given their respective quality scores). With identical quality scores, the payment rule reduced to simply paying the bid of the next highest bidder, and the format was dubbed the Generalized Second-Price Auction \citep{EOS2007,Varian2007}.

The Generalized Second-Price Auction format (GSP) has become enormously successful, becoming the de facto industry standard. Beyond Google, it has been deployed by most of the search engines worldwide. It is used by Amazon for its sponsored products offering (which is a major part of the company's \$31 billion per year, and rapidly growing, advertising business\footnote{\url{https://www.cnbc.com/2022/02/03/amazon-has-a-31-billion-a-year-advertising-business.html}}), and many other online retailers and marketplaces use it to power their ``Retail Media'' advertising offerings. DoorDash and Uber Eats use it to allow restaurants to advertise on their apps and websites.\footnote{\url{https://digiday.com/marketing/the-holy-grail-of-e-commerce-advertising-why-doordash-is-bolstering-its-advertising-offerings/}} Apple's App Store and Google's Play Store use it to let developers promote their apps. Travel sites like TripAdvisor, Expedia, Booking, and others use it to allow hotels to promote themselves to travelers. Job search sites like Indeed allow firms to post Sponsored Jobs that receive premium placement when job-seekers search for relevant terms. Instacart is using GSP to power its booming ad auction business.\footnote{\url{https://www.wsj.com/articles/instacart-goes-deeper-into-digital-advertising-as-grocery-delivery-slows-11630920600}}

Thus, until recently, it may have looked like the GSP had ``won,'' while GFP had fallen by the wayside, due to its poor equilibrium properties. In this paper, we show that this conclusion may be premature, and the situation is more nuanced. Specifically, we reconsider the question of equilibrium existence in the Generalized First-Price Auction when we include in the analysis a feature that is realistic (and common) in many real-world settings in which these types of auctions are used---stochastic quality scores. In the original \cite{EO2007} paper, quality scores were not considered (or equivalently, were constant and equal to 1), but the paper's conclusions on the non-existence of pure-strategy equilibria would have remained largely the same even if the scores were different from 1, but were constant. In practice, however, the scores are stochastic, and may vary widely from one impression to the next (e.g., one ad may perform particularly well with women in New York, while another ad for the same search term may perform particularly well with men in San Francisco). Other real-world features may also make the environment stochastic (e.g., bidders' campaigns may be ``paced'' if their daily budgets are too low).\footnote{See  \cite{AtheyNekipelov} and \cite{PinKey} for discussions and empirical evidence of randomness in quality scores and other auction parameters.} In such cases, the GFP strategy ``bid 1 cent higher than my opponent'' is in general no longer a best response, and the payoff functions become continuous, smooth, and better behaved. As a result, in such environments, we show that the Generalized First-Price Auction may in fact have a pure strategy equilibrium. Moreover, perhaps even more surprisingly, we show that the expected revenue under that equilibrium may substantially exceed that of the Generalized Second-Price Auction.\footnote{In Section 4, where we present this comparison, we explain why this finding does not contradict the celebrated Revenue Equivalence Theorem of \cite{myerson}, despite both formats resulting in identical allocations.} Together with other attractive features of the Generalized First-Price Auction, such as that it is easy to explain to the advertisers (and their managers) and that it does not suffer from potential auctioneer credibility problems \citep{AkbarpourLi}, our results suggest that it may deserve another look in some settings.

We start out by setting up the basic model to illustrate the above features (Section 2). The model is deliberately ``stripped down'' to its bare essence, to illustrate the key driving forces behind our results and make the derivations as transparent as possible. In particular, the model has only two bidders and two slots, and assumes that the distribution of quality scores is uniform on $[0,1]$. The simplicity of the model makes the revenue non-equivalence result particularly striking. After proving the existence and uniqueness of the equilibrium of GFP in this basic setting (Section 3) and comparing the revenues of GFP and GSP (Section 4), we explore a number of extensions and generalizations of the basic model.

In Section 5, we allow for more than two bidders and slots, and more general distributions of quality scores. In this more general setting, the pure-strategy equilibrium of GFP is not guaranteed to exist, but we show that there is only one ``candidate'' equilibrium, and provide a simple expression for computing it (Theorem 3). For any particular setting, one can then check whether the strategy provided by this candidate equilibrium expression is in fact the best response to all other bidders using that strategy, and thus determine whether the symmetric pure-strategy equilibrium exists. We also provide some sufficient conditions for equilibrium existence (Lemmas 1 and 2).

We then explore a variety of special cases of this framework: keeping the two-bidder, two-slot assumption but allowing for more general distributions of quality scores (Subsection 5.1), keeping the two-slot assumption but allowing both for $n \ge 2$ bidders and more general distributions (Subsection 5.2), and finally allowing for the general case of more than two slots, with the additional assumption of geometrically declining slot visibilities, going from top to bottom (Subsection 5.3). 

In Subsection 5.1.2, we discuss a particularly instructive special case that illustrates the role of quality score heterogeneity. Specifically, we maintain all of the simplifying assumptions of the basic model of Section 2, with only one change. Instead of assuming that quality scores are distributed uniformly on $[0,1]$, so that the ratio between the highest possible score and the lowest possible score is infinite, we ``shift'' the uniform distribution away from zero, so that the ratio is finite, so that there is less uncertainty about the ratio of the bidders' quality scores (and in the limit, the ratio of the highest possible quality score to the lowest possible quality score converges to 1, i.e., the no uncertainty case of \cite{EO2007}). 
Consistent with the intuition, we find that as the amount of uncertainty decreases, the range of other parameter values for which the pure-strategy equilibrium of GFP exists shrinks, completely disappearing in the limit as uncertainty goes away.\footnote{This is analogous to the finding that in the Edgeworth-Bertrand price competition context, a non-negligible amount of product differentiation is needed to restore the existence of a pure strategy equilibrium \citep{Benassy1989,Caplin1991,Canoy1996}. Of course, both the context of that setting and the mechanics of the results are very different from those in our paper.} Revenue comparisons between GSP and GFP also become more subtle as the amount of uncertainty is reduced. These findings suggest that the generalized first-price auction might be particularly attractive in settings in which the amount of uncertainty in quality scores is high (such as, e.g., ``display advertising'' types of settings, in which user intent is not clear and the ad system is trying to predict user interests based on their highly variable characteristics), whereas the generalized second-price auction might be relatively more attractive in settings in which the amount of uncertainty in quality scores is lower (such as, e.g., in pure ``sponsored search,'' in which ads are shown directly in response to user queries which clearly contain a lot of information and thus make the additional information about user characteristics relatively less important---thus making overall quality scores less variable). This comparative static is consistent with the current state of advertising at such search engines as Google, whose sponsored search system on Google.com (AdWords) is using the generalized second-price auction, while its ad system for external publishers, closer to display advertising (AdSense) has recently shifted to the generalized first-price auction.\footnote{\url{https://blog.google/products/adsense/our-move-to-a-first-price-auction/}}

In Section 6, we go beyond the assumption that all bidders are identical. In Subsection 6.1, we again go back to the most basic model of Section 2, but now assume that the two bidders' per-click values are heterogeneous. We find that even with this generalization, the equilibrium of the generalized first-price auction continues to exist. Moreover, equilibrium bids exhibit robust comparative statics with respect to key model parameters. The revenue comparison between GSP and GFP, however, is again more subtle. In Subsections 6.2 and 6.3, we allow bidders to possess private information at the time they submit their bids (about their values in 6.2 and about their quality scores 6.3), and prove that in these settings, the symmetric pure-strategy equilibrium is always guaranteed to exist, under general distributions and numbers of bidders and slots (Theorems 5 and 6). Section 7 concludes. 

\section{Basic Model}
\label{sec:basicmodel}
Two advertisers are competing for two advertising slots on an internet page. 
Each values a click from a user at \$1. The two slots vary in terms of visibility to the user.
The better slot has visibility normalized to $1$ and the lesser slot has visibility $\alpha\in \left( 0,1\right)$.

The bidders  simultaneously submit bids $\{ b_1,b_2\}$. Then a user arrives to the website. The website has information about user characteristics and the potential relevance of each ad to the user,  
and for each ad estimates the probability that the user will click on the ad conditional on noticing it. 
The probabilities for the two ads are (proportional to) $\{q_1,q_2\}.$ 
Each $q_i$ is drawn independently from the uniform distribution on $[0,1]$.\footnote{Since all payoffs in the model are proportional to the probability of a click, $q_i=1$ does not mean that there are users with an estimated probability of click equal to 1. Instead, whatever is the highest probability of click, we normalize it to be 1, without loss of generality.}

The auctioneer runs a generalized first-price auction (GFP) with scoring the bids
by $q_i$'s. 

That is, the publisher after observing $q_i$'s and bids, ranks the bids based on the scores:
\[
S_{i}=b_{i}q_{i}.
\]

The bidder with the highest score wins the top slot and, conditional on a click, pays his bid. The bidder with the second-highest score wins the lesser slot and likewise pays his bid if the user clicks on his ad.

The expected profits of the winner and loser, conditional on bid $b$ and realized $q$ are:
\begin{eqnarray*}
U_{W}(q,b) &=& q( 1-b), \\
U_{L}(q,b) &=&\alpha q(1-b).
\end{eqnarray*}

We are interested in characterizing the symmetric pure-strategy Nash equilibria of this game. From now on, when we talk about equilibrium (for example, the existence of equilibrium, revenues in equilibrium etc.), then unless we explicitly specify otherwise, we mean pure-strategy symmetric equilibrium.

\section{Existence of a Pure Strategy Equilibrium}
Our first main result is:
\begin{theorem}\label{Theorem1}
In the generalized first-price auction described above, there exists a unique symmetric pure-strategy equilibrium, with the equilibrium bids equal to $$b_1 = b_2 = \frac{2(1-\alpha)}{4-\alpha}.$$
\end{theorem}

Before we present the proof, we discuss the intuition. When $q$'s are known (or not used at all as in the original auctions), there is no pure-strategy equilibrium because at a tie, one of the bidders wants to either deviate to an $\epsilon$ higher bid or to zero.
Unequal bids cannot be an equilibrium either since then the loser wants to deviate to bid zero, the winner to almost match him, but then the loser should ``jump over.''

When bidders do not know $q$'s, the probability of winning changes smoothly in bid. Deviating to an $\epsilon$-higher bid changes the probability of getting the top position only slightly.
When there is enough uncertainty about $\frac{q_i}{q_j}$, the best response payoff becomes concave in own bid, and there exists a fixed point.\footnote{As we show in Section 5 when we consider a generalization of the basic model, concavity is important---continuity alone is not enough for existence.} 

The proof uses the fact that with the distribution of $q_i$'s having support starting at zero, the ratio of $q$'s ranges from $0$ to $\infty$. As we show in Section 5, it is possible to generalize the existence result for less variation in that ratio, but the existence of a pure-strategy equilibrium depends then on $\alpha$. The smaller $\alpha$ is, the smaller variation in the ratio is needed for existence. 

\begin{proof}
Since the optimization problem is symmetric for the two players, we consider only player~1.
Conditional on the two bids being \{$b_1,b_2\}$, the expected profit of bidder $1$ is
\begin{equation} \label{eq1-main}
\begin{split}
EU_1 (b_1,b_2) & =(1-b_1) E[q_1(\alpha+(1-\alpha) \1 _{b_1 q_1 > b_2 q_2})] \\
& = (1-b_1) \int_{0}^{1} q(\alpha +( 1-\alpha ) h(b_1,b_2,q)) dq, 
\end{split}
\end{equation}
where $h(b_1,b_2,q)$ is the probability that the score of bidder $1$ is higher than the score of bidder $2,$ given the
bids and $q_{1}=q$. 

In this expression, $\left( 1-b_{1}\right) $ is the expected profit
conditional on a click, $q\alpha $ is the probability of a click in position 
$2$ and the additional probability of a click $q\left( 1-\alpha \right) $ is
only conditional on winning the auction.

Player 1 wins whenever $b_1 q_1 > b_2 q_2$, so 
\[
h(b_1,b_2,q) = \min \left\{ 1, \frac{b_{1}}{b_{2}}q \right\}.
\]
Simplifying the expressions for expected payoffs we get: if $b_{1} \geq b_2$:
\begin{eqnarray}\label{eq2-main}
EU_{1}\left( b_{1},b_{2}\right)  = \left( 1-b_{1}\right) \left[ 
\frac{\alpha}{2}+\left( 1-\alpha \right) \frac{1}{6}\left( 3-\left( \frac{b_{2}}{
b_{1}}\right) ^{2}\right) \right],
\end{eqnarray}
and if $b_{1} \leq b_{2}$:
\begin{eqnarray}\label{eq3-main}
EU_{1}\left( b_{1},b_{2}\right)  = \left( 1-b_{1}\right) \left[  
\frac{\alpha}{2}+\left( 1-\alpha \right) \frac{1}{3}\frac{b_{1}}{b_{2}}\right] .
\end{eqnarray}

A necessary condition for equilibrium is that the FOC holds at the symmetric bidding profile $b_1=b_2=b$. 
By inspection, $EU_1(b_1,b_2)$ is differentiable in $b_1$ at symmetric bidding profiles and the derivative at such vectors is
\begin{eqnarray*}
&&\frac{\partial EU_1(b_1,b_2) }{\partial b_{1}} = \frac{2(1-\alpha) - b(4-\alpha)}{6b}.
\end{eqnarray*}
The unique solution of the FOC is hence a unique candidate for the pure symmetric equilibrium:

\begin{eqnarray*}
&&\frac{\partial EU_1(b_1,b_2)}{\partial b_{1}}=0 \text{ and } b_1=b_2=b  \iff  b_1 = b_2 = \frac{2(1-\alpha)}{4-\alpha}.\\
\end{eqnarray*}

To check that this is actually an equilibrium we need to check that these bids are global best responses. Direct inspection of (\ref{eq3-main}) shows that $EU_1(b_1,b_2)$ is concave in $b_1$ (since it is quadratic with a negative coefficient on $b_1$). Similarly, for any $b_1 \leq 1$ the expression in (\ref{eq2-main}) is concave in $b_1$ (note that $\frac{\partial^2 EU_1(b_1,b_2)}{\partial b_{1}^2}=\left( 1-\alpha \right)
\left( b_{1}-3\right) \frac{b_{2}^{2}}{3b_{1}^{4}}<0$ and bidding above $1$ is a dominated strategy).
Hence, the FOC is both necessary and sufficient for the optimality of the best response of bidder~1. 
\end{proof}

Given the equilibrium bid $b$, the equilibrium payoff is $$EU_1(b,b) =(1-b) \left( \frac{1}{6}\alpha +\frac{1}{3}\right). $$
The intuition is that the bidder gets payoff $(1-b)$ conditional on a click. They are half of the time the bidder with the higher score and half of the time with the lower score. Conditional on having the higher $q_i$, the expected $q_i$ is $\frac{2}{3}$, which is the expected probability of a click in the top position. Conditional on having the lower $q_i$, the expected $q_i$ is $\frac{1}{3}$, so the expected probability of a click on the lower position is $\alpha /3$. This yields the expression for the expected equilibrium number of clicks.

\section{Revenue Comparison: GFP vs. GSP}
Our second result compares bidding and revenues between the equilibrium of the GFP we characterized above and the equilibrium in a generalized second-price auction, GSP.

In the equilibrium of the GFP the expected revenue is $b(\frac{\alpha}{3} + \frac{2}{3})$ because the expected $q$ of the winner is $\frac{2}{3}$ and the expected $q$ of the loser is $\frac{1}{3}$. Plugging in the equilibrium bidding strategy we get
\begin{eqnarray}
REV_{GFP} =\frac{1}{3}\left( 1-\alpha \right) \frac{2\left( \alpha +2\right) }{4-\alpha}.
\end{eqnarray}

In a GSP, the loser pays $0$ and the winner $i$ pays 
\[
P=b_{j}\frac{q_{j}}{q_{i}}.
\]
That is the lowest amount the winner could have bid and still at least tied with the loser (with a bid $b_i=P$ the scores of the two bidders are equal).

With two ad slots and two bidders, the equilibrium bidding strategies in GSP are straightforward to find. Each bidder is guaranteed to win at least the second position and pay 0, so an equilibrium strategy is to bid per click the full value of the incremental expected number of clicks from the upgrade to the top position:
\[
b_{GSP}=1-\alpha.
\]
To see this, suppose bidder 1 happens to know $q_1$, $q_2$, and $b_2 = (1-\alpha)$. Then his options are: \begin{enumerate}
    \item Lose the auction, pay 0, and get the second position (and thus $\alpha q_1$ clicks), for the expected payoff of $\alpha q_1$; or 
    \item Win the auction, pay $b_2 q_2 = (1-\alpha) q_2$ in expectation, and get $q_1$ clicks, for the expected payoff of $q_1 - (1-\alpha) q_2$.
 \end{enumerate}

The second option is more attractive to bidder 1 whenever $q_1 > q_2$---and that is exactly the outcome that the bidder will get if he bids $b_1 = (1-\alpha)$.

The expected revenue in the GSP can be computed as: 

\[
REV_{GSP}=2\int_{0}^{1}q\left( \int_{0}^{q}\left( \left( 1-\alpha \right) \frac{
q_{2}}{q}\right) dq_{2}\right) dq=\frac{1}{3}\left( 1-\alpha \right). 
\]

This expression is intuitive: given the realized first and second order statistics of $q,$ $q^{\left( 1\right) }$ and $q^{(2)},$ respectively, only the winner pays and the expected payment of the winner per impression is 
\[
q^{\left( 1\right) } \cdot b\frac{q^{\left( 2\right) }}{q^{\left( 1\right) }}%
=bq^{\left( 2\right) }. 
\]%
The equilibrium bid is $b=\left( 1-\alpha \right) $ and the expected second-highest $q$ is $\frac{1}{3}.$

Finally, observe that $$REV_{GFP} = \frac{1-\alpha}{3} \frac{2\left( \alpha +2\right) }{4-\alpha} > \frac{1-\alpha}{3} = REV_{GSP},$$ because $\frac{2(\alpha +2)}{4-\alpha}>1$.

Summarizing the analysis above yields our second main result:

\begin{theorem}
The expected revenue in the Generalized First-Price Auction is equal to $$REV_{GFP} = \frac{2(1-\alpha)(\alpha + 2)}{3(4-\alpha)}.$$

The expected revenue in the Generalized Second-Price Auction is equal to $$REV_{GSP} = \frac{1-\alpha}{3}.$$

For all $\alpha \in (0,1)$, $$REV_{GFP} >  REV_{GSP}.$$

\end{theorem}

Note that as $\alpha$ gets close to $0$, the expected revenues from those two formats converge. But for all higher $\alpha$ the expected revenues are not the same, despite both formats resulting in efficient (and thus identical) equilibrium allocations. The difference can be substantial. For example, suppose $\alpha = 0.7$, which is a fairly typical ``dropoff'' value in ad auctions, meaning that the second slot gets 70\% as many clicks as the first one. Then the revenue in the generalized second-price auction is equal to 0.1, while the revenue for the generalized first-price auction is equal to $0.1 \cdot \frac{5.4}{3.3} \approx 0.164$---a difference of more than 60\%. 

\subsection{Understanding Revenue Non-Equivalence}

At first glance, this result on revenue differences seems to violate the Revenue Equivalence Theorem (RET) of \cite{myerson}---after all, in both formats, we end up with fully efficient (and thus identical) allocations: whichever bidder has the higher relevance $q_i$ gets the top position (because in both GFP and GSP, equilibria are symmetric, and the bidders are submitting identical bids). Since the value per click of each bidder is the same (equal to 1), that is the efficient allocation. And yet the revenues are different!

Of course, there is no contradiction. The purely technical reason is that RET holds subject to the payoffs of the lowest types being equal. And since in our model each bidder has only one informational type (the value per click of each bidder is fixed at 1, and relevance $q_i$ is not known to the bidder beforehand), which is thus automatically ``the lowest,'' RET has no bite. 

More substantively, the machinery behind the proof of RET in fact helps us understand the difference in expected revenues between the two formats. Specifically, consider a slightly modified version of the GFP ``auction.'' Bidder 2 has value $v_2=1$ and mechanically submits bid $b_2 = \frac{2(1-\alpha)}{4-\alpha}$ (the same as the equilibrium bid in Theorem 1 of Section 3). Bidder 1's value $v_1$ can take values between 0 and 1. The exact distribution of bidder 1's values is not important, but for concreteness, assume that it is uniform on $[0,1]$. Computing the equilibrium in this auction reduces to solving the single-agent bidding problem for every value $v_1 \in [0,1]$ of bidder 1. For $v_1 = 1$, we already know from Theorem 1 that the optimal best response $b_1(v_1)$ is equal to the bid of bidder 2, $\frac{2(1-\alpha)}{4-\alpha}$. For lower values $v_1 < 1$, we know from the standard monotonicity arguments that the optimal bid $b_1(v_1)$ must be weakly lower than $b_1(1) = b_2$. Thus, by the derivation parallel to that behind equation (\ref{eq3-main}) above, the expected utility of bidder 1 with value $v_1$ who submits bid $b_1$ is equal to  

\begin{eqnarray*}\label{rnet-main}
EU_{1}\left( v_1, b_{1}, b_{2}=\frac{2(1-\alpha)}{4-\alpha}\right)  &=& \left( v_1-b_{1}\right) \left[  
\frac{\alpha}{2}+\left( 1-\alpha \right) \frac{1}{3}\frac{b_{1}}{b_{2}}\right] \\
& = & \left( v_1-b_{1}\right) \left[  
\frac{\alpha}{2}+\left( 1-\alpha \right) \frac{1}{3}\frac{b_{1}}{\left(\frac{2(1-\alpha)}{4-\alpha}\right)}\right]\\
& = & \frac{1}{6} \left( v_1-b_{1}\right)\left[3\alpha + b_1\left(4-\alpha\right)\right].
\end{eqnarray*}
Since $b_1$ cannot be negative, and since $\alpha \in [0,1]$ (and thus $4-\alpha > 0$), this expression is maximized at the optimal bid
$$b_1(v_1) = \max\left\{0,\frac{v_1}{2}-\frac{3\alpha}{2(4-\alpha)}\right\}.$$
Thus, for small values of $v_1$, bidder 1 will optimally bid zero and settle for the second slot with no payment, while above the threshold value $v_1 = {3\alpha}/({4-\alpha})$, the bid is increasing linearly in~$v_1$. And for $v_1 =1$, as expected, we get $b_1(v_1 = 1) = \max\left\{0,\frac{1}{2}-\frac{3\alpha}{2(4-\alpha)}\right\} = \frac{2(1-\alpha)}{4-\alpha}$, which is equal to~$b_2$ and to the equilibrium bid of both bidders in the basic model of Sections 2 and 3.

Let $a^{GFP} (v_1)$ be the expected number of clicks that bidder 1 receives in this auction when he submits the optimal bid $b_1(v_1)$. That number is equal to $$ a^{GFP} (v_1) =   
\frac{\alpha}{2}+\left( 1-\alpha \right) \frac{1}{3}\frac{b_{1}(v_1)}{b_{2}}.$$

Consider now an analogous modified GSP auction, in which bidder 2's value is $v_2 = 1$ and his bid is fixed at $b_2^{GSP} = (1-\alpha)$. Bidder 1's value $v_1$ is again distributed uniformly on $[0,1]$. By the arguments analogous to those in the beginning of Section~4, the optimal bid for bidder 1 with value~$v_1$ is $b_1^{GSP}(v_1) = (1-\alpha)v_1$. And the expected number of clicks that bidder 1 receives in this auction when behaving in this optimal way is equal to $$ a^{GSP} (v_1) =   
\frac{\alpha}{2}+\left( 1-\alpha \right) \frac{1}{3}\frac{b^{GSP}_{1}(v_1)}{b^{GSP}_{2}}.$$

Note that in both GSP and GFP, the expected payoff of bidder 1 whose value is zero is also equal to zero. Thus, by the standard envelope argument underlying the Revenue Equivalence Theorem, we know that the expected payoff of bidder 1 whose value is 1 is equal to $\int_0^1 a^{GFP}(v_1)dv_1$ in the Generalized First-Price Auction and to $\int_0^1 a^{GSP}(v_1)dv_1$ in the Generalized Second-Price Auction. (And of course, these numbers are precisely the expected payoffs obtained by each bidder under these formats in the basic model in which both bidders' values are equal to~1.)

Now, the key observation is that for all $v_1$, we have $a^{GFP}(v_1)\le a^{GSP}(v_1).$ For the values of $v_1$ for which $b_1(v_1)=0$, this is immediate. For the remaining values of $v_1$, this inequality is equivalent to 
\begin{eqnarray*}
\frac{b_1(v_1)}{b_2} & \le & \frac{b^{GSP}_1(v_1)}{b^{GSP}_2}\\
& \Updownarrow & \\
\frac{\left(\frac{v_1}{2}-\frac{3\alpha}{2(4-\alpha)}\right)}{\left(\frac{1}{2}-\frac{3\alpha}{2(4-\alpha)}\right)} & \le & v_1\\
& \Updownarrow  & \\
{\frac{v_1}{2}-\frac{3\alpha}{2(4-\alpha)}} & \le & {\frac{v_1}{2}-v_1\frac{3\alpha}{2(4-\alpha)}},
\end{eqnarray*}
which is immediate. 

Thus, the expected payoff of bidder 1 whose value is 1 (and thus the expected payoff of each bidder in the basic model) is weakly lower in GFP than in GSP ($\int_0^1 a^{GFP}(v_1)dv_1 \le \int_0^1 a^{GSP}(v_1)dv_1$), and it is straightforward to further verify that the inequality is strict unless $\alpha = 0$ or $\alpha = 1$. But as we observed above, in the basic model both formats result in efficient allocations, and thus have identical total surpluses---and the expected revenue of the auctioneer, which is equal to the total surplus less the expected payoffs of the bidders, is therefore higher in the Generalized First-Price Auction than in the Generalized Second-Price Auction. 

To provide an illustration for the above arguments, consider a special case $\alpha = \frac{1}{2}$.  Using the above derivations, for this case, the optimal bid in GFP is $b^{GFP}_1(v_1) = 0$ for $v_1 \le \frac{3}{7}$ and $b^{GFP}_1(v_1) = \frac{v_1}2 - \frac{3}{14}$ for $v_1 > \frac{3}{7}$. The expected number of clicks for such a bidder is $\frac{1}{4}$ when $v_1 \le \frac{3}{7}$ (the bidder gets the second slot for sure) and $\frac{1}{4} + \frac{1}{6}\left(\frac{v_1 - 3/7}{1 - 3/7}\right)$ when $v_1 > \frac{3}{7}$ (there is a positive chance of getting the first slot). 

In GSP, the optimal bid is $b^{GSP}_1(v_1) = \frac{v_1}{2}$ for all $v_1$. The expected number of clicks is $\frac{1}{4} + \frac{1}{6}v_1$.

\begin{figure}
\centering
\includegraphics[width = 5in]{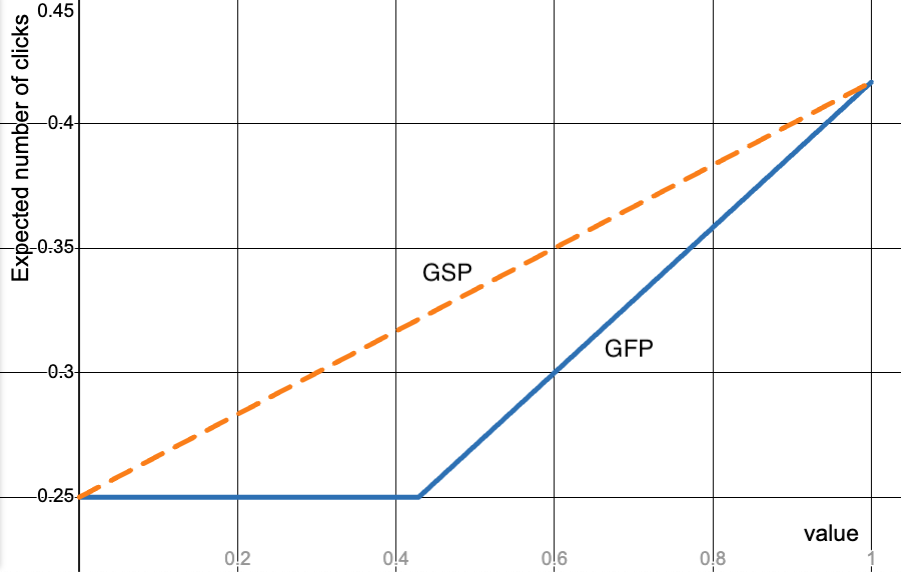} 
\caption{Expected number of clicks under GFP and GSP}
\label{envelope}
\end{figure}

Figure~\ref{envelope} plots the expected numbers of clicks under GFP and under GSP for $v_1 \in [0,1]$. For $v_1 = 0$, both values are the same---under both formats, the bidder gets the second slot for sure (and thus gets $\frac{1}{4}$ clicks in expectation). For $v_1 = 1$, both values are also the same---under both formats, the bidder bids the same amount as his or her opponent, and thus gets the top slot whenever his or her click-through rate is higher than that of the opponent (which of course is exogenous and is independent of the format). In that case, under both formats, the bidder gets $\frac{1}{4}+\frac{1}{6}=\frac{5}{12}$ clicks in expectation. The difference between the two formats is in the middle. Under GSP, the bidder bids truthfully for fraction $\frac{1}{2} = 1 - \alpha$ of a click (which is the marginal value of moving from the second slot to the first one), and so even for small values of $v_1$, the bid is positive and there is a positive chance of the bidder getting the top slot (when his or her click-through rate happens to be  high and the opponent's click-through rate happens to be low). The expected number of clicks is a linear function connecting the points $(0,\frac{1}{4})$ and $(1,\frac{5}{12})$. Under GFP, for small values of $v_1$, the bidder does not bid a positive amount---doing so would have little upside (very low chance of getting the top slot) and substantial downside (higher per-click payment even when getting the second slot). So for values below a threshold, $v_1^* = \frac{3}{7}$ in this case, the optimal bid is zero and the expected number of clicks won is flat at $\frac{1}{4}$. For values higher than $v_1^*$, the optimal bid becomes positive, there is a positive chance of getting a top slot, and the expected number of clicks starts rising linearly (at a steeper rate than under GSP) until getting to the same $(1,\frac{5}{12})$ point as in the graph of the expected number of clicks under GSP. Of course, it is immediate that for all intermediate values of $v_1$, the expected number of clicks won under GFP is lower than that under GSP. Thus the integral of the former function is strictly lower than the integral of the latter one, which by the arguments above (adapting the canonical \cite{myerson} reasoning) translates into the corresponding differences in expected bidder payoffs and expected auction revenues under the two formats. 

\section{Generalizations and Extensions: Identical Bidders}
\label{sec:ext-identical}

In this section, we explore a number of extensions and generalizations of the basic model of Section~\ref{sec:basicmodel}, maintaining the assumption that all bidders are identical: they have the same value $v = 1$ per click and do not have any private information.\footnote{In Section~\ref{sec:ext-heterog}, we study variations of the basic model in which bidders are heterogeneous.}

We allow a general number of bidders and slots and a general distribution of click probabilities. For this general case, we show that there exists at most one symmetric pure-strategy equilibrium of the Generalized First-Price Auction, and moreover, we show that a candidate for this equilibrium (or more precisely, for the equilibrium bid of each advertiser) is characterized by a linear equation. Thus, while we cannot guarantee existence for all combinations of numbers of bidders and slots and for all distributions of click-through rates (and as we show in subsections below, a symmetric equilibrium may indeed fail to exist in some cases), for any specific setting, our general result provides an approach for checking whether such an equilibrium exists. First, the solution of the linear equation is the unique candidate for the symmetric equilibrium bid $b^*$. Next, it is straightforward to check whether for a particular advertiser, if everyone else bids $b^*$, it is a best response to also bid~$b^*$ (e.g., Figure~\ref{checkingcandidate} in Subsection~\ref{lessheterog} below illustrates this ``candidate checking'' approach for a particular setting of that subsection). If it is a best response, then $b^*$ is the unique symmetric pure-strategy equilibrium; if it is not, then such an equilibrium does not exist. After stating and proving the general result, we explore a number of special cases in more detail. 

Formally, suppose there are $n\ge 2$ bidders, each with value $v=1$ per click, and $n$ slots, with visibilities $1 = \alpha_1 \ge \alpha_2 \ge \dots \ge \alpha_n \ge 0$ , with at least one inequality strict (several of the slots can have visibility zero, allowing for the possibility that there are fewer positive-visibility slots than bidders). The click-through rate of each bidder $i$, $q_i$, is drawn independently from distribution $F$ on  $[\underline{q}, \overline{q}]$ that has a probability density function $f$ that is continuous everywhere except perhaps at a finite number of points and $\int q^2f(q)^2dq<\infty$.\footnote{The lower bound $\underline{q}$ is non-negative. The upper bound $\overline{q}$ can be equal to $+\infty$, in which case we assume that distribution $F$ has a finite expectation.} 

\begin{theorem}
The Generalized First-Price Auction in the above setting has at most one symmetric pure-strategy equilibrium $b^*$.
\end{theorem}
\begin{proof}
Consider bidder $i$, and suppose all the other bidders bid $b^*$.  Note that in any symmetric equilibrium, $b^*\in(0,1)$. Bids $b^*>1$ would imply negative payoffs. If $b^*=1$, then bidding slightly below $1$ would be a profitable deviation. If $b^*=0$, then bidding slightly above zero would be a profitable deviation.

So consider any $b^*\in(0,1)$, and let bidder $i$ submit some bid $b_i$. The expected number of clicks that $i$ receives depends only on the ratio $r:=\frac{b_i}{b^*}$, because if both $b_i$ and $b^*$ are multiplied by the same positive number, the expected number of clicks that $i$ receives remains unchanged. Denote by $C\left(\frac{b_i}{b^*}\right)$ the expected number of clicks that $i$ receives. 

The expected number of clicks of bidder $i$ as a function of the bid ratio is:
\[
  C(r)=\int_{\underline{q}}^{\overline{q}} q\,\Psi\big(F(rq)\big)f(q)\,dq,
  \qquad
  \Psi(s):=\sum_{j=1}^{n}\alpha_{j}\binom{n-1}{j-1}s^{\,n-j}(1-s)^{\,j-1},
\]
where $\Psi(s)$ is the expected visibility of a bidder who outranks each of his opponents
independently with probability $s$.\footnote{We use the conventions that $0^{0}=1$ (for the terms with $j=n$
and $F(\cdot)=0$, and with $j=1$ and $1-F(\cdot)=0$), and that $F(\cdot)=0$ below the support and
$F(\cdot)=1$ above it.} 
Differentiating,
\[
  C'(r)=\int_{\underline{q}}^{\overline{q}} q^{2}\,\Psi'\big(F(rq)\big)f(rq)f(q)\,dq,
\]
which is finite and continuous in $r.$ So, $C$ is continuously differentiable on $(0,\infty)$ and  
non-decreasing. Moreover, because $\Psi'\ge0$, and at
least one of the inequalities $\alpha_{1}\ge\cdots\ge\alpha_{n}$ is strict, $0<C'(1)<\infty$.\footnote{Together with $C(1)>0$, this implies that the candidate bid derived below satisfies $b^*\in (0,1)$. Finally, $C(r)>0$ for
every $r>\rho:=\underline{q}/\overline{q}$. Moreover, $C(r)>0$ for all of $r\in(0,\infty)$ whenever $\alpha_{n}>0$ or
$\underline{q}=0$, which is the case in every specific setting considered below.
}

The expected payoff of bidder $i$ is equal to $$(1-b_i)C\left(\frac{b_i}{b^*}\right).$$ Thus, the first-order condition for the optimal best response is given by 
$$-C\left(\frac{b_i}{b^*}\right) + \frac{1-b_i}{b^*} C'\left(\frac{b_i}{b^*}\right) = 0.$$
In a symmetric equilibrium, when $b_i = b^*$, this expression simplifies to 
$$-C\left(1\right) + \frac{1-b^*}{b^*} C'\left(1\right) = 0,$$ which is equivalent to the linear equation $-b^*C(1) + (1-b^*)C'(1)=0$, giving us a formula for the unique candidate bid profile:
$$b^* = \frac{C'(1)}{C'(1)+C(1)}.$$

\end{proof}

If all opponents bid $b^*$, we can express bidder $i$'s strategy as choosing $r \in [0, 1/b^*]$ and bidding $b_i = rb^*$ (bidding more than the top of this range is dominated since the value of a click is $1$; bidding less than 0 is not allowed). With this notation, the best response of bidder $i$ is to maximize
$$ U(r) = (1-rb^*)C(r).$$ 

To determine whether the pure-strategy equilibrium exists, it suffices to check whether $U(r)$ is globally maximized at $r=1.$ 

We can provide the following sufficient conditions:

\begin{lemma}[Single-crossing sufficient condition]
\label{lemma1SC}
Suppose $C(r)>0$ on $(0,1/b^*]$, and define
\[
M(r):=\frac{C'(r)}{C(r)+rC'(r)}.
\]
If $M(r)$ is weakly decreasing on $(0,1/b^*]$, then $r=1$ is a global maximizer of $U(r)$ on
$[0,1/b^*]$ and $b^*$ is a symmetric pure-strategy equilibrium of the GFP.

\end{lemma}

\begin{proof}
Differentiate \(U(r)=(1-rb^*)C(r)\):
\[
U'(r)=-b^*C(r)+(1-rb^*)C'(r).
\]
Rearranging,
\[
U'(r)
=
\bigl(C(r)+rC'(r)\bigr)
\left(
\frac{C'(r)}{C(r)+rC'(r)}-b^*
\right).
\]
Thus
\[
U'(r)=\bigl(C(r)+rC'(r)\bigr)\bigl(M(r)-b^*\bigr).
\]
Because
\[
b^*=\frac{C'(1)}{C(1)+C'(1)}=M(1),
\]
the stated single-crossing condition implies
\[
U'(r)\ge0\quad\text{for }r<1,
\qquad
U'(r)\le0\quad\text{for }r>1.
\]

Hence $U(r)\le U(1)$ for every $r\in(0,1/b^*]$, and since $U$ is continuous at $r=0$, $r=1$ maximizes $U$ on $[0,1/b^*]$.

\end{proof}

\begin{lemma}[Convexity of $1/C(r)$ sufficient condition]
\label{Lemma2}
Suppose $C(r)>0$ and $1/C(r)$ is convex on $r\in(0,1/b^*]$. Then $r=1$ globally maximizes $U(r)$ on
$[0,1/b^*]$, and hence $b^*$ is a symmetric pure-strategy equilibrium of the GFP. Moreover,
convexity of $1/C(r)$ is implied by log-concavity of $C(r)$.
\end{lemma}

\begin{proof}

By convexity of $1/C(r)$, the graph of $1/C(r)$ lies above its tangent at $r=1$:
\[
\frac{1}{C(r)} \ge \frac{1}{C(1)} -\frac{C'(1)}{C(1)^2}(r-1)
=
\frac{C(1)+(1-r)C'(1)}{C(1)^2}.
\]
By definition of $b^*$, for $r\in[0,1/b^*)$,
\[ 
C(1)+(1-r)C'(1)>0,
\]
so
\[
C(r)
\le
\frac{C(1)^2}{C(1)+(1-r)C'(1)}.
\]
Using again the definition of $b^*$, we have
\[
1-rb^*
=
\frac{C(1)+(1-r)C'(1)}{C(1)+C'(1)}.
\]
Therefore,
\[
U(r)
=
(1-rb^*)C(r)
\le
\frac{C(1)^2}{C(1)+C'(1)}.
\]
Since 
\[
U(1) = (1-b^*)C(1)  = \frac{C(1)^2}{C(1)+C'(1)},
\]
we get 
\[
U(r)\le U(1)
\]
for every $r\in(0,1/b^*)$. At the boundary $r=1/b^*$, we have $U(r)=0\le U(1)$ (and since $U$ is continuous we also get $U(0) \le U(1)$). Thus $r=1$ is a global maximizer of $U(r)$.

Finally, if $C(r)$ is log-concave then $-\log C(r)$ is convex, and since $x\mapsto e^{x}$ is convex
and increasing, $1/C(r)=e^{-\log C(r)}$ is convex as well.
\end{proof}

In the rest of this section, we explore a variety of special cases of the above setting. In Subsection~\ref{subsec:n2k2}, we consider the case of $2$ bidders and $2$ slots. In Subsection~\ref{subsec:nge3k2}, we consider the case of two slots with positive visibilities and more than $2$ bidders.  Finally, in Subsection~\ref{subsec:kge3}, we consider the case of $n\ge 3$ slots and bidders. 

\subsection{Two Slots, Two Bidders}
\label{subsec:n2k2}
Assume there are two bidders and two slots. Normalize the visibilities to be $1$ and $\alpha$. 

In this case, $C(r)$ simplifies to:
$$ C(r) = \int_{\underline{q}}^{\overline{q}} q \left[F(rq) + \alpha(1-F(rq))\right] f(q) dq = \alpha \mathbb{E}[q] + (1-\alpha) \mathbb{E}[qF(rq)]. $$

This allows us to establish existence of pure strategy equilibria in GFP for a wide range of distributions:
\begin{corollary}
\label{C1}
In the case of two slots and two bidders, if the distribution of $q$ has support starting at $\underline q =0 $ and a weakly decreasing density, then the GFP has a unique symmetric pure strategy equilibrium. 
\end{corollary}

\begin{proof}
Extend \(f\) by setting \(f(x)=0\) for \(x>\overline q\). Then almost everywhere:
\[
        C'(r)=(1-\alpha) J(r),
        \qquad
        \text{where }  J(r):=\int_0^{\bar q} q^2 f(rq)f(q)\,dq \geq 0.
\]
Because \(f\) is weakly decreasing, for every \(s>r\) and every \(q\geq 0\), $f(sq)\leq f(rq)$. 
Hence
\[
        J(s)\leq J(r).
\]
Thus $C'(r)$ is weakly decreasing. Since $C(r)>0$, this implies that $1/C(r)$ is convex and by Lemma \ref{Lemma2} the unique candidate $b^*$ is an equilibrium of the GFP.

\end{proof}

Examples of such distributions include the uniform distribution, an exponential distribution, a truncated normal distribution $N(\mu,\sigma^2)$ with truncation at zero and any $\mu \leq 0$, a Beta$(a,b)$ distribution with any $0<a \leq1$ and $b\geq 1$, a Gamma$(k, \theta)$ distribution for any $k\in(0,1]$, a Weibull$(k,\lambda)$ distribution for any $k\in(0,1]$. The Lomax (i.e., Pareto Type II) distribution $F(q)=1-(1+\frac{q}{\lambda})^{-\theta}$ for $\lambda>0$ and $\theta >2$ so that the mean and variance are finite, also satisfies the conditions.

\paragraph{Revenue Comparisons.}
In the two-bidder case, we can say more about revenue comparisons. Define
\[
L=\E[\min\{q_1,q_2\}],\qquad H=\E[\max\{q_1,q_2\}],
\]
and
\begin{equation}
\label{eq:D-general}
D=\int q^2f(q)^2\dd q.
\end{equation}
This allows us to express the expected number of clicks as:
\begin{equation}
\label{eq:C1-Cp1}
C(1)=\frac{H+\alpha L}{2},\qquad C'(1)=(1-\alpha)D,
\end{equation}
and the (candidate) equilibrium bid in the GFP as:
\begin{equation}
\label{eq:b-star-HLD}
b^*=\frac{(1-\alpha)D}{\frac{H+\alpha L}{2}+(1-\alpha)D}.
\end{equation}

In the GFP both bidders pay $b^*$ per click and the total expected number of clicks is $H+\alpha L$, so the expected GFP revenue is:
\begin{equation}
\label{eq:Rev-GFP}
REV_{GFP}=b^*(H+\alpha L)
=\frac{(1-\alpha)D(H+\alpha L)}{\frac{H+\alpha L}{2}+(1-\alpha)D}.
\end{equation}

In the GSP, the equilibrium bid is $1-\alpha$ and the winner pays the minimum per-click bid needed to tie the loser's score (and the loser pays nothing). Therefore, the expected GSP revenue is:
\begin{equation}
\label{eq:Rev-GSP}
REV_{GSP}=(1-\alpha)L.
\end{equation}

The revenue comparison is
\begin{equation}
\label{eq:comparison-ineq}
REV_{GFP} > REV_{GSP}
\quad\Longleftrightarrow\quad
D(H-L+2\alpha L)>\frac{L(H+\alpha L)}{2}.
\end{equation}
Define the cutoff 
\begin{equation}
\label{eq:alpha_threshold_gen}
\widehat\alpha
=\frac{\frac{LH}{2}-D(H-L)}{2DL-\frac{L^2}{2}}.
\end{equation}

Then
\begin{equation}
\label{eq:threshold-rule}
REV_{GFP}>REV_{GSP}
\quad\Longleftrightarrow\quad
\alpha>\widehat\alpha.
\end{equation}
If $\widehat\alpha\le0$, GFP dominates for every $\alpha\in(0,1)$.\footnote{It can be shown that the denominator in $\widehat \alpha$ is positive and hence the sign depends only on the sign of the numerator.} 
If $\widehat\alpha\ge1$, GSP dominates for every $\alpha\in(0,1)$. Therefore:

\begin{corollary}
In the case of two slots and two bidders, assuming existence of a pure strategy symmetric equilibrium in the GFP, if expected revenue in the GFP is higher than in the GSP for some $\alpha$, then it is also higher for all larger values of $\alpha$.
\end{corollary}

To finish this section we consider a few specific distributions and discuss existence beyond the conditions of Corollary \ref{C1}.

\subsubsection{Exponential and Power Distributions}
For concrete examples, consider the exponential and power distributions. 
For the exponential distribution, $F(q)=1-e^{-\lambda q},$ we can compute $$ C(r) = \frac{1}{\lambda} \left(1 - \frac{1-\alpha}{(1+r)^2}\right).$$
That implies that the equilibrium bid in the GFP is
$$b^*=\frac{1-\alpha}{4}.$$
Moreover, the order statistics for this distribution are $H=\frac{3}{2\lambda}$, $L=\frac{1}{2\lambda}$. The expected revenues in the two auction formats are:

$$REV_{GFP}=b^*(H+\alpha L)=\frac{(1-\alpha)(3+\alpha)}{8\lambda},$$

$$REV_{GSP}=(1-\alpha)L=\frac{(1-\alpha)}{2\lambda}.$$

This shows that in the case of \emph{exponential distributions}, for all $\lambda$ and $\alpha$ the GSP yields a higher expected revenue. 

Now consider the power distributions, $F(q)=q^\theta$, for $q\in[0,1]$ and $\theta \in (0,1]$.
Then
\[
\E[q]=\frac{\theta}{\theta+1}, \qquad H=\frac{2\theta}{2\theta+1}, \qquad  L=\frac{2\theta^2}{(\theta+1)(2\theta+1)}.
\]

Also,
\[
D=\int_0^1q^2\theta^2q^{2\theta-2}\dd q
=\frac{\theta^2}{2\theta+1}.
\]
Substituting these expressions into (\ref{eq:alpha_threshold_gen}) yields $\widehat\alpha=0$. Hence, in the case of \emph{power distributions}, for all $\theta \in (0,1]$ and all $\alpha \in (0,1)$ the GFP yields higher expected revenue.

\subsubsection{Uniform Distribution with Less Heterogeneity in Click-Through Rates}
\label{lessheterog}
To finish the discussion of the two-bidder case, we return to the uniform case, but allow for a smaller range of heterogeneity in click-through rates. In our baseline model we have assumed that $q_i$'s are drawn from a uniform distribution on $[0,1]$. That implies that the range of possible ratios of $\frac{q_1}{q_2}$ is $[0,\infty)$. What if the range of possible ratios is less extreme? 

To this end, assume that the probabilities of clicks are proportional to $q_i$ which are drawn from a uniform distribution on $[u,u+1]$, for some parameter $u \geq 0$. When $u=0$ we have as a special case our previous model and as $u$ gets larger, there is less and less ex-ante uncertainty about the ratio of the CTRs.\footnote{Note that we can normalize the CTRs by dividing them by $u+1$, so that the model is equivalent to assuming that the distribution of the CTRs is uniform $[\frac{u}{u+1}, 1]$. As $u\rightarrow \infty$, the uncertainty about the CTRs disappears.}  

\paragraph{Equilibrium Existence.}

Let $\mu=\mathbb E[q]=u+\frac12$. Substituting the uniform distribution into our formula for $C(r)$ we get
\[
C(r)=
\begin{cases}
\alpha\mu,
& 0\le r\le \dfrac{u}{u+1},\\[1.1em]
\alpha\mu+(1-\alpha)
\dfrac{-3u(u+1)^2+2(u+1)^3r+u^3r^{-2}}{6},
& \dfrac{u}{u+1}\le r\le 1,\\[1.1em]
\alpha\mu+(1-\alpha)
\dfrac{3u^3+3(u+1)^2-2u^3r-(u+1)^3r^{-2}}{6},
& 1\le r\le \dfrac{u+1}{u},\\[1.1em]
\mu,
& r\ge \dfrac{u+1}{u}\quad (u>0).
\end{cases}
\]

Moreover, we can compute 
\[
L=\mathbb E[\min\{q_1,q_2\}]=u+\frac13,
\qquad
H=\mathbb E[\max\{q_1,q_2\}]=u+\frac23,
\]
and
\[
D=\int_u^{u+1}q^2f(q)^2\,dq
  =\int_u^{u+1}q^2\,dq
  =u^2+u+\frac13 .
\]
Therefore,
\[
C(1)=\frac{H+\alpha L}{2}
     =\frac{u+\frac23+\alpha\left(u+\frac13\right)}{2},
\]
and
\[
C'(1)=(1-\alpha)D
      =(1-\alpha)\left(u^2+u+\frac13\right).
\]

The candidate symmetric equilibrium bid is therefore
\[
b^{*}
=
\frac{(1-\alpha)\left(u^2+u+\frac13\right)}
{
\frac{u+\frac23+\alpha\left(u+\frac13\right)}{2}
+(1-\alpha)\left(u^2+u+\frac13\right)
}
=
(1-\alpha)
\frac{6u(u+1)+2}
{3u(1+2u)(1-\alpha)+6u+4-\alpha}.
\]
Recall 
\[
U(r)=(1-rb^*)C(r).
\]
By construction, \(U'(1)=0\).  
We consider all possible deviations. 

A deviation to a zero bid yields
\[
U(0)=\alpha\mu=\alpha\left(u+\frac12\right).
\]
At the candidate equilibrium,
\[
U(1)
=
(1-b^*)C(1)
=
\frac{C(1)^2}{C(1)+C'(1)}.
\]
Hence the zero-bid deviation is not profitable iff
\[
U(1)\ge U(0),
\]
or equivalently,
\[
C(1)^2-\alpha\left(u+\frac12\right)\bigl(C(1)+C'(1)\bigr)\ge 0.
\]
Substituting \(C(1)=(H+\alpha L)/2\) and \(C'(1)=(1-\alpha)D\), this becomes
\[
\left(\frac{H+\alpha L}{2}\right)^2
-
\alpha\left(u+\frac12\right)
\left[
\frac{H+\alpha L}{2}+(1-\alpha)D
\right]
\ge 0.
\]
Using the values of \(H,L,D\) for this distribution, the left-hand side simplifies to
\[
\frac{(1-\alpha)
\left[
9u^2+12u+4
-\alpha(36u^3+45u^2+21u+4)
\right]}{36}.
\]
Thus the zero-bid deviation is not profitable iff
\[
\alpha\le \alpha^*
\equiv
\frac{9u^2+12u+4}
{36u^3+45u^2+21u+4}.
\]

For \(u>0\), the other endpoint deviation is to bid enough to win for sure,
i.e., \(r_+=(u+1)/u\).  Since \(C(r_+)=\mu\), direct substitution gives
\[
U(1)-U(r_+)
=
\frac{(1-\alpha)
\left[
27u^3+45u^2+28u+6
-\alpha(9u^3+6u^2+u)
\right]}
{6u\left[6u^2+9u+4-\alpha(6u^2+3u+1)\right]}
\ge 0.
\]
Therefore, the endpoint \(r_+\) is never a profitable deviation.

The piecewise expression for \(C(r)\) also implies that on the two interior regions,
\[
\frac{r^4}{1-\alpha}U''(r)
=
\begin{cases}
u^3-\dfrac{b^{*}u^3}{3}r
-\dfrac{2b^{*}(u+1)^3}{3}r^4,
& \dfrac{u}{u+1}<r<1,\\[1.2em]
\dfrac{2b^{*}u^3}{3}r^4
+\dfrac{b^{*}(u+1)^3}{3}r
-(u+1)^3,
& 1<r<\dfrac{u+1}{u}.
\end{cases}
\]
The first expression is decreasing in \(r\), while the second is increasing in \(r\).  Hence $U'(r)$ cannot cross zero from above in this region and so it is sufficient to check deviations at the endpoints. In particular, any
profitable deviation below \(r=1\) would imply that the zero-bid deviation is profitable as well. Analogously, any
profitable deviation above \(r=1\) would imply that the endpoint \(r_+\) is profitable.  Consequently,
the candidate bid is a global best response if and only if the deviation to zero bid is not profitable, i.e., $
\alpha\le \alpha^*$.

\paragraph{Revenue comparisons.}
Using the revenue formulas described above, we get for the uniform distribution:
\[
REV_{GFP}
=
b^{*}
\left[
u+\frac23+\alpha\left(u+\frac13\right)
\right],
\]
\[
REV_{GSP}
=
(1-\alpha)\left(u+\frac13\right).
\]

Substituting further and simplifying we get that within the existence region
\(\alpha\le \alpha^*\),
\[
REV_{GFP}
>
REV_{GSP}
\quad\Longleftrightarrow\quad
\alpha>\widehat\alpha,
\]
where
\[
\widehat\alpha
=
\frac{u(u+1)}
{(3u+1)(4u^2+3u+1)}.
\]
Moreover, for every \(u>0\),
\[
\alpha^*-\widehat\alpha
=
\frac{
4(2u+1)(3u^2+3u+1)^2
}
{
(3u+1)(4u^2+3u+1)(36u^3+45u^2+21u+4)
}
>0,
\]
so the interval \((\widehat\alpha,\alpha^*]\) is nonempty.  For \(u=0\),
\[
\widehat\alpha=0,
\qquad
\alpha^*=1,
\]
recovering the baseline result that GFP revenue exceeds GSP revenue for all
\(\alpha\in(0,1)\).

We summarize our findings for this distribution as:
\begin{corollary}
Suppose the click probabilities are proportional to $q_{i}$ drawn from the uniform distribution on
$[u,u+1]$. Then the generalized first-price auction has a symmetric pure-strategy equilibrium if and
only if
\[
\alpha\le\alpha^{*}\equiv\frac{9u^{2}+12u+4}{36u^{3}+45u^{2}+21u+4},
\]
in which case the equilibrium is unique. 
Moreover, whenever the equilibrium exists, $REV_{GFP}>REV_{GSP}$ if and only if
$\alpha>\widehat\alpha=\frac{u(u+1)}{(3u+1)(4u^{2}+3u+1)}$, and the interval
$\left(\widehat\alpha,\alpha^{*}\right]$ is nonempty for every $u\ge0$.
\end{corollary}

\begin{figure}
\centering
\begin{tabular}{cc}
\subfloat[$u = 0$]{\includegraphics[width = 2.5in]{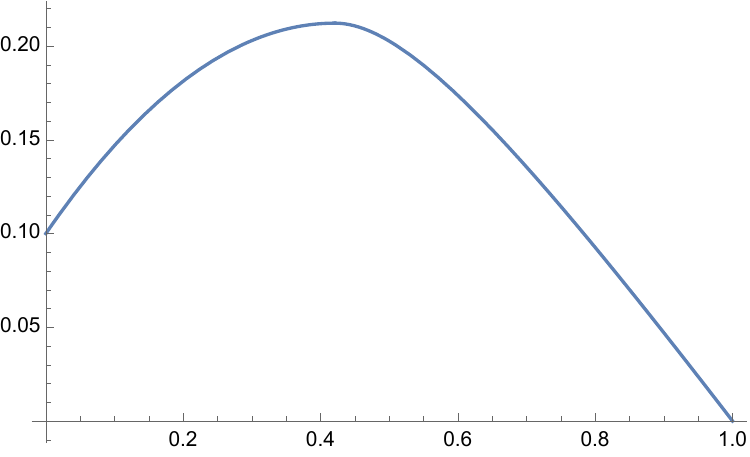}} &
\subfloat[$u = 0.5$]{\includegraphics[width = 2.5in]{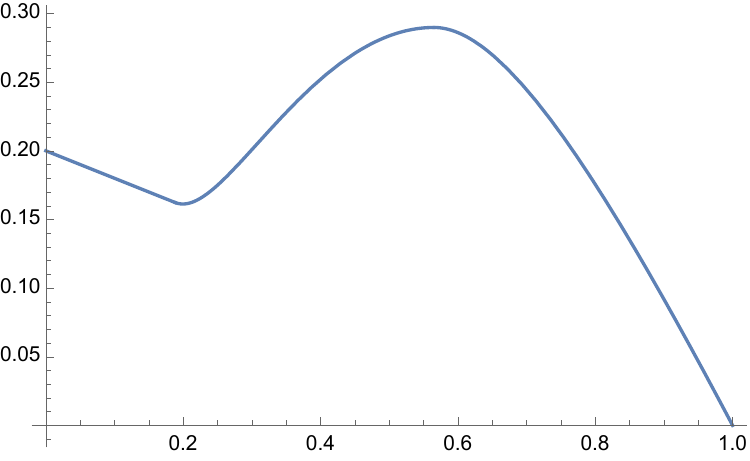}}\\[0.3in]
\subfloat[$u = 1$]{\includegraphics[width = 2.5in]{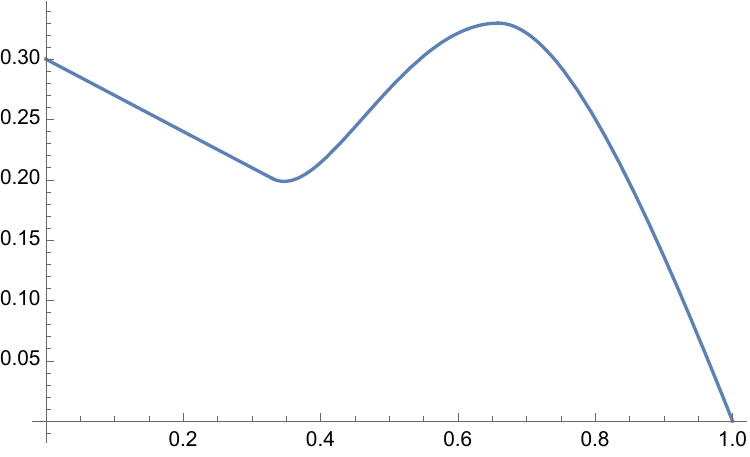}} &
\subfloat[$u=2$]{\includegraphics[width = 2.5in]{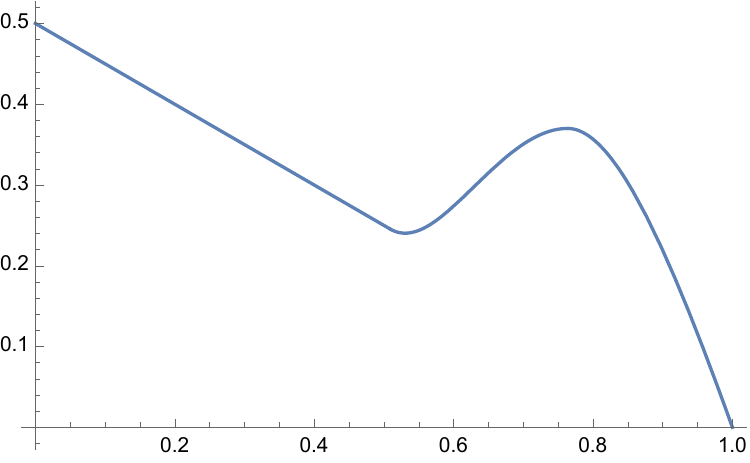}} \\[0.3in]
\end{tabular}
\caption{Expected payoffs from responding to a candidate equilibrium bid}
\label{checkingcandidate}
\end{figure}

To get the intuition for the reasons behind pure-strategy equilibrium existence (or non-existence) for different levels of uncertainty in click-through rates, in Figure~\ref{checkingcandidate} we present graphs of expected payoffs of bidder 1 as a function of her own bid (varying from 0 to 1). We set $\alpha = 0.2$ and vary the level of uncertainty from the maximum possible in our model ($u=0$) to lower ones ($u=2$, corresponding to the case in which the largest click-through rate is $1.5$ times the smallest, and two values in between, $u=0.5$ and $u=1$). In each panel, we assume that bidder 1's opponent, bidder 2, bids according to the ``candidate'' equilibrium bid $b^*$ that satisfies the first-order condition, and thus to check whether the ``candidate'' is indeed an equilibrium, we need to verify that the local maximum that corresponds to it in the expected payoff graph is in fact a global one. For the case of maximum uncertainty, $u=0$, the graph is simple: it is concave, and the global optimality (and thus equilibrium existence) is immediate. Once the amount of uncertainty is reduced, however, the picture becomes more subtle. Below a certain level (specifically, below $b^*\frac{u}{u+1}$), it does not make sense for bidder 1 to bid any positive amount: she is guaranteed to get the second slot (because even if she gets the highest possible score and her opponent gets the lowest one, she would still not be able to outbid him for the top slot), and so any positive bid amount simply reduces her payoff from winning that second slot. So the real test of global optimality is whether for bidder 1, the expected payoff from bidding $b^*$ (and thus trying to win the top slot with probability 50\%) is higher than the payoff from bidding zero and settling for the second slot at no cost. For intermediate values of uncertainty in click-through rates ($u=0.5$ and $u=1$) that is indeed the case, and so both bidders bidding $b^*$ is a pure-strategy equilibrium. However, once this uncertainty is reduced further ($u=2$), bidder 1 is better off bidding zero instead of bidding $b^*$, and a pure-strategy equilibrium no longer exists. Of course, in the limit, as $u \rightarrow +\infty$, we converge back to the deterministic case of no uncertainty in click-through rates \citep{EO2007}, in which a pure-strategy equilibrium likewise does not exist.

\subsection{Two Slots, More than Two Bidders}
\label{subsec:nge3k2}
In this subsection, assume that there are two slots with positive visibilities $\alpha_1=1$, $\alpha_2=\alpha$, and $n>2$ bidders.  

With $2$ slots (with positive visibilities) and more than $2$ bidders, the expression for $C(r)$ changes because a bidder can now receive the top slot, the second slot, or no slot at all. Moreover, a deviation to a zero bid now yields a zero payoff and is therefore never profitable. If all bidders' $q_i$'s were identical and commonly known, the equilibrium would simply be for all bidders to bid $1$ and make no profits. With uncertain $q_i$'s, however, bidding $b=1$ is not an equilibrium, and the existence of a symmetric pure-strategy equilibrium is not guaranteed---in fact, even if the distribution satisfies the conditions of Corollary~\ref{C1}, existence is not guaranteed.

For revenue comparisons with more than two bidders, we note that it is easier to compare expected revenues in the GFP with those in a VCG auction than with those in the GSP (since the equilibrium in the VCG is easier to characterize than the equilibrium in the GSP) and so we present these comparisons.

\subsubsection{Power Distributions}
Consider the distributions 
\[
F(q)=q^\theta,
\qquad q\in[0,1],
\qquad \theta>0.
\]
Note that we now allow $\theta>1$ so we are no longer constraining the family to have decreasing density.

Let $m:=n-1$ be the number of competitors to bidder $i$. Conditional on own quality \(q\), bidder \(i\)'s probability of receiving the top slot if she bids $r$ is
\[
F(rq)^m,
\]
and her probability of receiving the second slot is
\[
m(1-F(rq))F(rq)^{m-1}.
\]
Thus
\[
C(r)
=
\int_0^1
q\left[
F(rq)^m+\alpha m(1-F(rq))F(rq)^{m-1}
\right]f(q)\,dq.
\]

\begin{proposition}[Existence for the power distributions]
If the click-through rates are distributed according to
\[
F(q)=q^\theta,\qquad q\in[0,1],\qquad \theta>0,
\]
then the function \(M(r)\) defined in Lemma~\ref{lemma1SC} is weakly decreasing.  Hence the GFP with two slots has a symmetric pure-strategy equilibrium for every \(n>2\) and every \(\alpha\in(0,1)\). 
\label{Prop1}
\end{proposition}

The proof is in the Appendix. It derives $M(r)$ explicitly and shows that it is decreasing, and then applies Lemma \ref{lemma1SC} to establish existence.

We can also calculate the equilibrium bid explicitly. We get (using the formulas derived in the proof of Proposition \ref{Prop1}):
\[
C(1)
=
\frac{
\theta\left(1+(n-1)(1+\alpha)\theta\right)
}{
\left((n-1)\theta+1\right)(n\theta+1)
},
\]
and
\[
C'(1)
=
\frac{
(n-1)\theta^2\left((n-1)\theta+1-\alpha(\theta+1)\right)
}{
\left((n-1)\theta+1\right)(n\theta+1)
},
\]
therefore
\[
b^*
=
\frac{
(n-1)\theta\left((n-1)\theta+1-\alpha(\theta+1)\right)
}{
\left((n-1)\theta+1\right)^2-(n-1)\alpha\theta^2
}.
\]

In the special case of the uniform distribution, $\theta=1$, the equilibrium bid simplifies to
    $$b^*=\frac{n-2\alpha}{\frac{n^2}{n-1}-\alpha}.$$

\paragraph{Revenue Comparisons.}
Note that with three or more bidders, the equilibrium of the Generalized \emph{Second-}Price Auction (GSP) is not immediate in our setting. However, we can compare the resulting revenue of GFP to the expected revenue in the Vickrey-Clarke-Groves (VCG) auction. Let
\[
q_{(1)}\ge q_{(2)}\ge q_{(3)}
\]
denote the first, second, and third highest order statistics.

\paragraph{GFP revenue.}
In the symmetric GFP equilibrium, all bidders use the same bid \(b^*\), so the allocation is efficient.  Total expected clicks equal
\[
n C(1).
\]
Therefore, the expected GFP revenue is
\[
REV_{GFP}
=
b^* n C(1)
=
n\frac{C(1)C'(1)}{C(1)+C'(1)}.
\]
Substituting the formulas above,
\[
REV_{GFP}
=
\frac{
n(n-1)\theta^2
\left(1+(n-1)(1+\alpha)\theta\right)
\left((n-1)\theta+1-\alpha(\theta+1)\right)
}{
\left((n-1)\theta+1\right)(n\theta+1)
\left(\left((n-1)\theta+1\right)^2-(n-1)\alpha\theta^2\right)
}.
\]

\paragraph{VCG revenue.}
  With two slots, the expected VCG revenue is
\[
REV_{VCG}
=
(1-\alpha)\mathbb E[q_{(2)}]
+
2\alpha \mathbb E[q_{(3)}].
\]
For the power distribution family,
\[
\mathbb E[q_{(2)}]
=
\frac{
n(n-1)\theta^2
}{
(n\theta+1)((n-1)\theta+1)
},
\]
and
\[
\mathbb E[q_{(3)}]
=
\frac{
n(n-1)(n-2)\theta^3
}{
(n\theta+1)((n-1)\theta+1)((n-2)\theta+1)
}.
\]
Hence
\[
REV_{VCG}
=
(1-\alpha)
\frac{
n(n-1)\theta^2
}{
(n\theta+1)((n-1)\theta+1)
}
+
2\alpha
\frac{
n(n-1)(n-2)\theta^3
}{
(n\theta+1)((n-1)\theta+1)((n-2)\theta+1)
}.
\]
Equivalently,
\[
REV_{VCG}
=
\frac{
n(n-1)\theta^2
\left((n-2)\theta+1+\alpha((n-2)\theta-1)\right)
}{
(n\theta+1)((n-1)\theta+1)((n-2)\theta+1)
}.
\]

\paragraph{Revenue ranking.}
Subtracting the expressions for the VCG revenue from the GFP revenue yields
\[
REV_{GFP}-REV_{VCG}
=
\frac{
n(n-1)\alpha\theta^3\left(1-(n-1)\alpha\right)
}{
\left((n-2)\theta+1\right)
\left((n-1)\theta+1\right)
\left(\left((n-1)\theta+1\right)^2-(n-1)\alpha\theta^2\right)
}.
\]
It is straightforward to verify that the denominator is strictly positive.  Thus the sign of the difference in expected revenues is determined entirely by the sign of $\left(1-(n-1)\alpha\right)$. Therefore,
\begin{proposition}
In the case of 2 slots, $n>2$ bidders and the power distributions of $q_i$'s, the expected revenue from the GFP is higher than the expected revenue from the VCG if and only if $\alpha<\frac{1}{n-1}$.   
\end{proposition}

Thus, the revenue ranking is independent of the power parameter of the distribution. The parameter \(\theta\) affects the levels of revenue, but not the sign of the revenue difference.

\subsubsection{Exponential Distributions}
Continuing with two slots and more than two bidders, we now consider the case of exponential distribution of click-through rates. 

\begin{proposition}[Existence for the exponential distributions]
\label{Prop_large_n_expon}
If the click-through rates are distributed according to the  exponential distribution  $F(q)=1-e^{-\lambda q},$
then the GFP has a symmetric pure-strategy equilibrium for every \(n>2\), every $\lambda>0$, and every \(\alpha\in(0,1)\).
\end{proposition}

The proof is in the Appendix and it proceeds by showing that in the case of exponential distributions, $C(r)$ is log-concave, and hence by Lemma \ref{Lemma2}, the unique candidate $b^*$ is an equilibrium.

\paragraph{Revenue Comparisons.}
We can also compare the expected revenues between the GFP and the VCG. We show that, unlike in the power-distribution case, the VCG always dominates in the exponential case:

\begin{proposition}
For 2 slots and $n>2$ bidders, if the distribution of $q_i$'s is exponential, then for any $\alpha \in (0,1)$ the expected revenue from the GFP is smaller than the expected revenue from the VCG.
\label{Prop_revenue_exponential}
\end{proposition}

\subsection{More than Two Slots}
\label{subsec:kge3}
We now discuss the case of more than 2 slots. In general, since we need to consider the probability of landing in any of the slots, the formula for $U(r)$ can get quite complicated. The most direct way to verify existence in any particular instance is to compute $U(r)$ and check numerically if $r=1$ is the global maximum. 

To develop some intuition, however, one canonical case turns out to be tractable. Suppose there are $n\ge 2$ bidders and $n$ slots with visibilities: 
\[
\alpha_j=\alpha^{j-1},\qquad j=1,\ldots,n,
\]
where $\alpha\in(0,1)$.

In that case, the expression for $C(r)$ is relatively tractable. Recall $m:=n-1,$ and suppose all other bidders bid $b^*$, while bidder $i$ bids $rb^*$. Conditional on own quality $q$, the number $\tilde M$ of opponents ranked above bidder $i$ is binomial with success probability $1-F(rq)$. Since slot $j+1$ has visibility $\alpha^j$,
\[
\mathbb E[\alpha^{\tilde M}\mid q]
=
\left[F(rq)+\alpha(1-F(rq))\right]^m.
\]
Assuming that the range of $q$'s is $[0,1]$, bidder $i$'s expected clicks are
\begin{equation}
C(r)
=
\int_0^1
q\left[\alpha+(1-\alpha)F(rq)\right]^m f(q)dq,    
\label{n_slots_Cr}
\end{equation}
where $F(rq)$ is understood to equal $1$ when $rq\ge 1$.

This observation simplifies calculations of $U(r)$. To illustrate, we show existence for the power distributions with $\theta \leq 1$ (for $\theta>1$ sometimes the pure-strategy symmetric equilibrium does not exist). 

\begin{proposition}[Power distributions of qualities and geometric visibilities]
Suppose there are $n$ slots and bidders, and slot visibilities decline geometrically  $\alpha_j=\alpha^{j-1}$ for some $\alpha \in (0,1)$. Suppose click-through rates  are distributed according to  
 \[
F(q)=q^\theta,\qquad q\in[0,1],
\]
for some $0<\theta\le 1$. Then $b^*$ is a symmetric pure-strategy equilibrium of the GFP.
\label{Prop_more_than_2_slots}
\end{proposition}
\medskip

The proof is in the Appendix. It establishes the claim by showing that $1/C(r)$ is convex and applying Lemma \ref{Lemma2}.

We can also verify, by tedious algebraic calculations, that the equilibrium exists in case the click-through rates  have exponential distributions by also showing that $1/C(r)$ is convex (and then applying Lemma \ref{Lemma2}). 

Finally, via direct calculations, we can compare expected revenues in the GFP and VCG. These calculations show that, as for $n=2$, the GFP dominates in the case of the power distributions and the VCG dominates in the case of the exponential distributions. We omit these calculations since they do not bring any major new intuition beyond what we have already learned in the $n=2$ case.

\section{Generalizations and Extensions: Heterogeneous Bidders}
\label{sec:ext-heterog}
In this section, we consider variations of the basic model that relax the assumption that all bidders are identical at the moment they submit their bids. 

In Subsection~\ref{subsec:vV}, we consider the case in which one bidder has per-click value $v>0$ and the other bidder has per-click value $V>v$. In Subsection~\ref{subsec:Bayesian}, we consider the case in which each bidder's per-click value $v_i$ is independently drawn from some distribution $G$ and is that bidder's private information. Finally, in Subsection~\ref{subsec:BayesianQS}, we consider another variation, going back to the base case in which each bidder's per-click value is equal to $v=1$, but now allowing each bidder $i$ to observe his quality score $q_i$ before submitting his bid.

\subsection{Unequal Per-Click Values $v>0$ and $V \ge v$}
\label{subsec:vV}

Consider the following variation of the basic model of Section~\ref{sec:basicmodel}. Quality scores are still distributed uniformly on $[0,1]$, but the per-click values of the two bidders are no longer necessarily identical. The value of the second bidder, $v_2$, is equal to some value $v > 0$, while the value of the first bidder, $v_1$, is equal to $V \ge v$.

In this asymmetric setting, the existence result continues to hold. 

\begin{theorem}
In the generalized first-price auction described above, there exists a unique pure-strategy equilibrium.
\end{theorem}
\begin{proof}
Conditional on the two bids being $\{b_1,b_2\}$, the expected profit of bidder 1 is:
\begin{align*}
    E U_{1}\left(b_{1}, b_{2}\right)&=\left(V-b_{1}\right) E\left[q_{1}\left(\alpha+(1-\alpha) 1_{b_{1} q_{1}>b_{2} q_{2}}\right)\right].
\end{align*}
whereas the profit for bidder 2 is:
\begin{align*}
    E U_{2}\left(b_{1}, b_{2}\right)&=\left(v-b_{2}\right) E\left[q_{2}\left(\alpha+(1-\alpha) 1_{b_{1} q_{1}<b_{2} q_{2}}\right)\right].
\end{align*}
As before, if we simplify the expressions for expected payoffs, we get that if $b_1\geq b_2$:
\begin{align}
    E U_{1}\left(b_{1}, b_{2}\right)&=\left(V-b_{1}\right)\left[\frac{\alpha}{2}+(1-\alpha) \frac{1}{6}\left(3-\left(\frac{b_{2}}{b_{1}}\right)^{2}\right)\right]\label{eq1}\\
    E U_{2}\left(b_{1}, b_{2}\right)&=\left(v-b_{2}\right)\left[\frac{\alpha}{2}+(1-\alpha) \frac{1}{3} \frac{b_{2}}{b_{1}}\right],\label{eq2}
\end{align}
whereas if $b_1\leq b_2$:
\begin{align}
    E U_{1}\left(b_{1}, b_{2}\right)&=\left(V-b_{1}\right)\left[\frac{\alpha}{2}+(1-\alpha) \frac{1}{3} \frac{b_{1}}{b_{2}}\right]\label{eq3}\\
    EU_{2}\left(b_{1}, b_{2}\right)&=\left(v-b_{2}\right)\left[\frac{\alpha}{2}+(1-\alpha) \frac{1}{6}\left(3-\left(\frac{b_{1}}{b_{2}}\right)^{2}\right)\right].\label{eq4} 
\end{align}
If in equilibrium $b_1\geq b_2$ (as will be the case), then equations (\ref{eq1}) and (\ref{eq2}) are the relevant ones. 

Advertiser 1's problem is:
\begin{align*}
    \max_{b_1}\left(V-b_{1}\right)\left[\frac{\alpha}{2}+(1-\alpha) \frac{1}{6}\left(3-\left(\frac{b_{2}}{b_{1}}\right)^{2}\right)\right]
\end{align*}
with FOC:
\begin{align}
\frac{-3b_{1}^{3}+b_{2}^{2}\left( 2V-b_{1}\right) \left(1-\alpha \right)}{6b_{1}^{3}}   =0. \label{foc1}
\end{align}

The sign of this derivative  depends on the sign of the numerator and that changes once from positive to negative. Hence, for every $b_2$ there is a unique best response $b_1$ that is a solution to: 
\begin{align}
-3b_{1}^{3}+b_2^{2}\left( 2V-b_{1}\right) \left( 1-\alpha \right)=0
\end{align}

We can solve for the inverse of the best response:
\begin{eqnarray*}
    b_2 &=&\sqrt{\frac{3b_1^{3}}{\left( 1-\alpha \right) \left(2V-b_1\right) }}.
\end{eqnarray*}

Advertiser 2's problem is:
\begin{align*}
    \max_{b_2}\left(v-b_{2}\right)\left[\frac{\alpha}{2}+(1-\alpha) \frac{1}{3} \frac{b_{2}}{b_{1}}\right]    
\end{align*}
with FOC:
\begin{align}
    \frac{ 2\left( 1-\alpha \right) \left( v-2b_{2}\right)
-3\alpha b_1}{6b_{1}} =0.\label{foc2}
\end{align}

This derivative is decreasing in $b_2$. So the best response of bidder 2 is 
\begin{align}
b_2=\max \left\{0, \frac{1}{2}v-\frac{3\alpha b_{1}}{4\left( 1-\alpha \right) } \right\}.\label{BR2}
\end{align}

The equilibrium bid $b_1^*$ is the solution to:
\begin{equation}
\sqrt{\frac{3b_{1}^{3}}{\left( 1-\alpha \right) \left( 2V-b_{1}\right) }}%
=\max \left\{ \frac{2v\left( 1-\alpha \right) -3\alpha b_{1}}{4\left(
1-\alpha \right) },0\right\}.\label{bidding_asym}
\end{equation}
The function on the left is increasing and the function on the
right is decreasing in $b_1$. At $b_1=0$ the LHS is $0$ and the RHS is
positive. At $b_1=V$ the ranking is the opposite. 
So, there is a unique solution and at the solution $b_1>0$. 
Given the unique solution for $b_1^*$, there is a unique solution for $b_2^*$ as well (that can be found using (\ref{BR2})). That pair $(b_1^*,b_2^*)$ is thus the unique equilibrium for the case $b_1 \geq b_2$.

To rule out the cases in which $b_1 < b_2$ in equilibrium, we repeat the above computation for equations (\ref{eq3}) and (\ref{eq4}), and find that for all possible pairs of values $v<V$ and all possible values of $\alpha$, all candidate solutions have  $b_1> b_2$, leading to a contradiction. This completes the proof.
\end{proof}

In addition to showing equilibrium existence and uniqueness, we can also use the expressions in the proof to establish the comparative statics of equilibrium bids with respect to bidder values. 

First, observe that equilibrium bids 
$(b_1^*, b_2^*)$ are both increasing in $v$ (as long as $v$ remains smaller than $V$). To see this, note that the RHS of equation (\ref{bidding_asym}) increases in $v$, so the solution $b_1$ increases in $v$. Moreover, the LHS is the equilibrium bid $b_2$, and that is increasing in $b_1$. 

Second, the comparative statics of equilibrium bids with respect to the higher value, $V$, are more subtle. Again from equation (\ref{bidding_asym}), we observe that if $V$ increases, equilibrium bid $b_1^*$ has to increase as well. To determine the change in the other bid, $b_2^*$, we now consider equation (\ref{BR2}), and observe that as $b_1^*$ increases, $b_2^*$ has to \emph{decrease}. In words, if the value of the weaker bidder increases, both bidders increase their equilibrium bids. By contrast, if the value of the stronger bidder increases, then that bidder's equilibrium bid also goes up, but the weaker bidder's equilibrium best response goes down. 

Finally, we can also show that both equilibrium bids $(b_1^*, b_2^*)$ are decreasing in $\alpha$: as the second position becomes more valuable, the competition for the first one becomes less intense. To see this, rewrite equation (\ref{bidding_asym}) as 
\begin{equation}
\sqrt{\frac{3b_{1}^{3}}{\left( 1-\alpha \right) \left( 2V-b_{1}\right) }} + 
\frac{3\alpha b_{1}}{4\left(
1-\alpha \right) } = \frac{1}{2}v.\label{bidding_asym_alt}
\end{equation}
(We can do this because in equilibrium, the ``$\max$'' expression on the RHS of equation (\ref{bidding_asym}) is always equal to its first term, and can thus be simplified.) Next, observe that the expression on the LHS of equation (\ref{bidding_asym_alt}) is increasing in $b_1$ and also increasing in $\alpha$, and since this expression has to be constant (equal to the RHS), as $\alpha$ increases, $b_1$ has to decrease. To show that $b_2$ also has to decrease in $\alpha$, note that by construction, $b_2$ is equal to the first of the two terms in the sum on the LHS of equation (\ref{bidding_asym_alt}). Since the sum of the two terms is constant, if we want to show that the first one of them is decreasing in $\alpha$, it is sufficient to show that \emph{the ratio} of the first term to the second term is decreasing in $\alpha$. Dropping constant terms, this ratio is proportional to \begin{equation}\sqrt{\frac{(1-\alpha)b_1}{\alpha^2 (2V-b_1)}}\label{ratio},
\end{equation}
and since we've already established that $b_1$ is decreasing in $\alpha$, it is immediate by inspecting the terms under the square root that the expression in equation (\ref{ratio}) is also decreasing in $\alpha$.

Next, given equilibrium values ($b_1^*,b_2^*$), we can compute  equilibrium payoffs:
\begin{align*}
    E U_{1}\left(b_1^*, b_2^*\right)&=\left(V-b_1^*\right)\left[\frac{\alpha}{2}+(1-\alpha) \frac{1}{6}\left(3-\left(\frac{b_2^*}{b_1^*}\right)^{2}\right)\right]\\
    E U_{2}\left(b_1^*, b_2^*\right)&=\left(v-b_2^*\right)\left[\frac{\alpha}{2}+(1-\alpha) \frac{1}{3} \frac{b_2^*}{b_1^*}\right]
\end{align*}
and the expected revenue for GFP:
\begin{align*}
    REV_{GFP}&=\int_{0}^{1} \int_{q_{2} \frac{b_{2}^*}{b_1^*}}^{1}\left(q_{1} b_{1}^*+\alpha q_{2} b_{2}^*\right) d q_{1} dq_{2}+\int_{0}^{1} \int_{0}^{q_{2}\frac{b_{2}^*}{b_1^*}}\left(q_{2} b_{2}^*+\alpha q_{1} b_{1}^*\right) d q_{1}dq_2\\
    &=\frac{1}{2}b_{1}^{*}+\frac{1}{2}\alpha b_{2}^{*}+\frac{1-\alpha }{6}\frac{(b_{2}^{*})^{2}}{b_{1}^{*}}
    .
\end{align*}
As before, for GSP, the bidding strategies are  straightforward to find. Specifically, it is easy to see that advertiser 1 should bid $b_1=V(1-\alpha)$, while advertiser 2 should bid $b_2=v(1-\alpha)$. The argument is identical to the one in Section 4. To compute the expected revenue in GSP, it is convenient to normalize $V=1$ (and $v\le V=1$), and the revenue is then:
\begin{align*}
    REV_{GSP}&=\int_{0}^{1} \int_{q_{2} v}^{1}(1-\alpha) v\frac{q_2}{q_1} q_{1} d q_{1} d q_{2}+\int_{0}^{1} \int_{0}^{q_{2} v}(1-\alpha)\frac{q_1}{q_2} q_{2} d q_{1} d q_{2}\\
    &=\frac{(1-\alpha)v}{2}-\frac{(1-\alpha)v^2}{6}.
\end{align*}

The relationship between the revenues in GFP and GSP is now more complex. Figure~\ref{heteroval} plots the revenues in the two auction formats for $V=1$ and the values of $v$ between 0 and 1. We show three cases illustrating a wide range of possible outcomes: $\alpha \in\{0,0.4,0.8\}$. As we saw earlier, the revenue in GFP is greater than in GSP for $v=1$ and any $\alpha>0$ (they are equal for $\alpha=0$ and $v=1$). 
However, for lower values of $v$, that is not always the case. For instance, when $
\alpha$ is low (see the graphs for $\alpha=0$), for relatively high (but lower than 1) values of $v$, the revenue in GSP is higher, and the ranking reverses for small $v$. For intermediate values of $\alpha$ (e.g., $\alpha=0.4$),  GFP dominates for all $v$. Finally, for large $\alpha$  (see the graph for $\alpha=0.8$), GFP dominates for large $v$, and the ranking flips for small $v$ (but the differences are minimal for small $v$ and large $\alpha$, so in the bottom-right panel we zoom to a range of small $v$ to make the difference in revenues visible).

\pagebreak

\begin{landscape}

\begin{figure}
\centering
\begin{tabular}{ccc}
\subfloat[$\alpha = 0$, Auction Revenues]{\includegraphics[width = 2.5in]{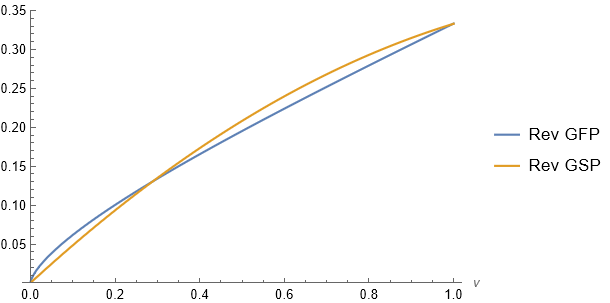}} &
\subfloat[$\alpha = 0.4$, Auction Revenues]{\includegraphics[width = 2.5in]{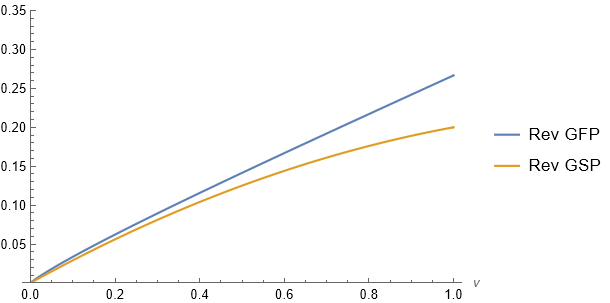}} &
\subfloat[$\alpha = 0.8$, Auction Revenues]{\includegraphics[width = 2.5in]{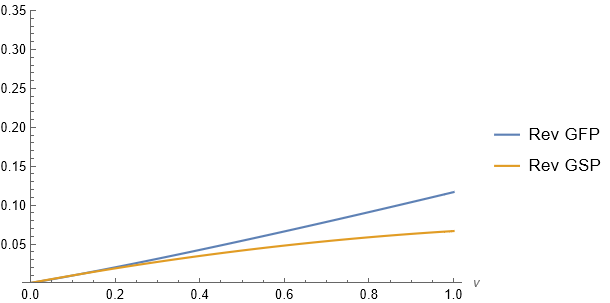}}\\[0.5in]
\subfloat[$\alpha = 0$, Revenue Difference]{\includegraphics[width = 2.5in]{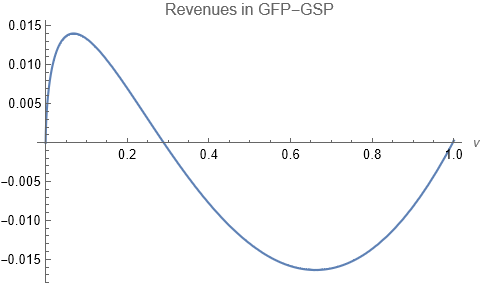}} &
\subfloat[$\alpha = 0.4$, Revenue Difference]{\includegraphics[width = 2.5in]{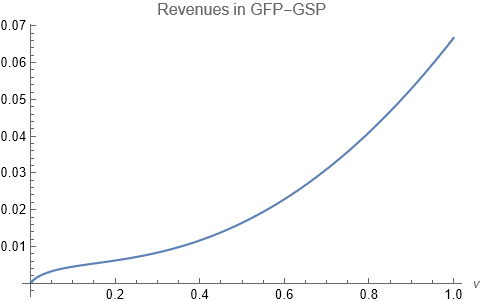}} &
\subfloat[$\alpha = 0.8$, Revenue Difference]{\includegraphics[width = 2.5in]{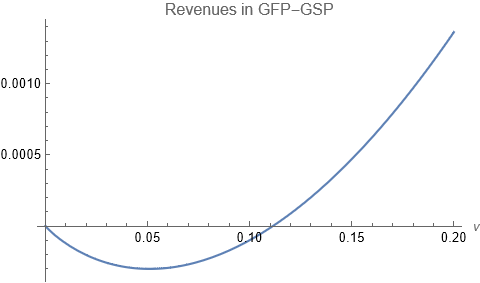}}\\[0.5in]
\end{tabular}
\caption{Revenues in GFP and GSP Auctions}
\label{heteroval}

\end{figure}
\end{landscape}
\newpage

\subsection{Privately Observed Per-Click Values}
\label{subsec:Bayesian}

In the preceding sections, including the asymmetric model of Subsection~\ref{subsec:vV}, we assumed that bidders' values were commonly known, and correspondingly the solution concept we used was the pure-strategy Nash Equilibrium: knowing the bid(s) of the other bidder(s), any particular bidder $i$ does not have an incentive to change his own. As we discussed in the Introduction, given the dynamic nature of advertising auctions, and the fact that bidders have many opportunities to revise their bids, this is a logical solution concept to apply: a profile of bids that is not a pure-strategy Nash Equilibrium will not be a stable outcome in this setting, since at least one bidder will always have an incentive to revise his or her bid. Nevertheless, both for completeness and because in some settings it may be a more applicable assumption, it is interesting to explore how our results extend to incomplete-information settings (with the corresponding solution concept of Bayesian Nash Equilibrium). We explore this question in the current subsection and the subsequent one. In the current subsection, we assume that each bidder's value $v_i$ is drawn independently from some distribution $G$, and is that bidder's private information. We prove that in this setting, a symmetric Bayesian Nash Equilibrium is always guaranteed to exist. In Subsection~\ref{subsec:BayesianQS}, we go back to the case of each bidder's value $v_i$ being equal to 1, and assume instead that each bidder's quality score is his or her private information. 

Consider a variation of the general case of Section~\ref{sec:ext-identical}, with $n\ge 2$ bidders and $n$ slots, with visibilities $1=\alpha_1 \ge \alpha_2 \ge \cdots \ge \alpha_n \ge 0$. As in that setting, the click-through rate of each bidder $i$, $q_i$, is drawn independently from the same distribution $F$ on $[\underline{q}, \overline{q}]$, and is not observed by the bidder at the time he submits his bid. What the bidder does observe is his value $v_i$, drawn independently from other random variables in the model, from some distribution $G$ with continuous and positive density $g$ on an interval $[0,\overline{v}]$. A symmetric equilibrium in this setting can be represented by a bidding function $\beta(v)$ such that for each bidder $i$, for every realized per-click value $v_i$, it is optimal in expectation to submit the bid $\beta(v_i)$ into the Generalized First-Price Auction in this setting, conditional on other bidders also using the bidding function $\beta(\cdot)$.

\begin{theorem}
A symmetric equilibrium bidding function $\beta(\cdot)$ is guaranteed to exist.
\end{theorem}
\begin{proof}
The proof proceeds by discretizing the space of types and actions; showing that for each discretization, a symmetric (mixed-strategy) Bayesian-Nash Equilibrium exists and possesses certain properties; and then taking a limit of discretization and showing that a limiting bidding function is in fact a symmetric equilibrium of the original game.

Formally, for each integer $m\ge 1$, consider an auction game $\Gamma^{m}$ that is analogous to the game above, except that the space of allowed bids is restricted to $\left\{0,\tfrac{\overline{v}}{m},\tfrac{2\overline{v}}{m},\dots,\overline{v}\right\}$, and each bidder $i$ has $m$ possible values $v_i$, with the $j$-th possible value being equal to $E_G\lbrack v|\frac{j-1}{m}\overline{v}\le v \le \frac{j}{m}\overline{v}\rbrack$ and the probability of the bidder having that value being equal to $G(\frac{j}{m}\overline{v})-G(\frac{j-1}{m}\overline{v})$.

Since $\Gamma^m$ is a symmetric game, it has a symmetric equilibrium. For $j=1, \dots, m$ and $k=0,\dots,m$, let $p(m,j,k)$ denote the probability that when a player's realized value is the $j$-th possible one, $E_G\lbrack v|\frac{j-1}{m}\overline{v}\le v \le \frac{j}{m}\overline{v}\rbrack$, he bids $\frac{k}{m}\overline{v}$ in this symmetric equilibrium. (If for some $m$, there are multiple symmetric equilibria, pick any one of them.)

Of course, this equilibrium can be in mixed strategies, but nevertheless, it possesses an important structural property that is analogous to the standard equilibrium monotonicity property of continuous auction settings. Specifically, in the current discrete setting, the highest bid that a type $j$ submits with a positive probability has to be less than or equal to the lowest bid that any higher type $j'>j$ submits. This is because the expected allocation and the expected auction payment depend only on the bid (and not on the type), and the marginal value of the additional clicks from bidding a higher amount is strictly higher for the bidder with the higher value.

As $m$ increases, the bid of the lowest type has to converge to zero (because in equilibrium, bid never exceeds value). Define $\beta(0)=0$. 

Next, for each $m$, let $\overline{b}^m$ be the highest bid amount that is submitted in equilibrium with a positive probability. By the structural monotonicity property above, this bid is submitted by the bidder of the highest possible type. Since this bid always belongs to the compact interval $[0,\overline{v}]$, there must exist an infinite subsequence of indices $m$ along which $\overline{b}^m$ converges to some limit point~$\overline{b}$. Let $\beta(\overline{v})=\overline{b}$.

Next, we construct $\beta({\frac{\overline{v}}{2}})$ in the same way: from this infinite subsequence of indices $m$ that we selected above, further identify an infinite sub-subsequence for which the highest bid that occurs with positive probability for a bidder in a neighborhood of $\frac{\overline{v}}{2}$ converges to some limit point, and define $\beta({\frac{\overline{v}}{2}})$ as that limit point. 

Next, construct $\beta({\frac{\overline{v}}{4}})$, then $\beta({\frac{3\overline{v}}{4}})$, and so on for all fractions with powers of 2 in the denominator. For the remaining values $v$, define $\beta(v)$ as the supremum of bids $\beta(v')$ for all lower types $v'<v$. It is immediate that the resulting function $\beta(\cdot)$ is weakly monotone. Moreover, it has to be continuous because an infinitesimally small increase in the number of clicks associated with going from $\beta(v)$ to $\beta(v+\epsilon)$ cannot justify a discrete, non-infinitesimal increase in the bid (and thus in the expected GFP payment). The fact that $\beta(v)$ is a symmetric Bayesian Nash Equilibrium now follows from standard limit arguments. 
\end{proof}

\subsection{Privately Observed Quality Scores}
\label{subsec:BayesianQS}

In our baseline model and its variations considered above, we assumed that when submitting their bids, bidders have the same amount of uncertainty about their own ad quality scores for the upcoming impression and about their competitors' quality scores. We now consider a variation of the baseline model in which bidder $i$ is more informed about the quality score of his own ad than about the quality score of bidder $j$, and vice versa. 

Formally, we consider a variation of the general case of Section~\ref{sec:ext-identical}, with $n\ge 2$ bidders, each with value $v=1$, and $n$ slots, with visibilities $1=\alpha_1 \ge \alpha_2 \ge \cdots \ge \alpha_n \ge 0$. As in that setting, the click-through rate of each bidder $i$, $q_i$, is drawn independently from the same distribution $F$ on $[\underline{q}, \overline{q}]$. 

What is different from the general case of Section~\ref{sec:ext-identical} is that now each bidder $i$ knows his click-through rate $q_i$ before submitting his bid in the generalized first-price auction. 

The main result of this section is that in this setting, there always exists a unique symmetric equilibrium bidding function $\beta(q)$.

To see this, start with the following closely related setting. Suppose the auctioneer does not know the bidders' click-through rates, and runs the standard Vickrey-Clarke-Groves mechanism. I.e., she asks each bidder $i$ to report his private information (i.e., his click-through rate $q_i$), allocates the slots to the bidders in the descending order of reported click-through rates, and charges each bidder $j$ the amount equal to the reduction in the expected total number of clicks received by other bidders due to bidder $j$'s presence. (Since by assumption, per-click values of all bidders are normalized to 1, the expected number of clicks is equal to the expected value of clicks.) By the standard VCG logic, in this setting, it is optimal for each bidder to report his quality score truthfully. 

Let $\nu(q)$ denote the expected number of clicks received by a bidder with quality score $q$ in the above VCG auction. Let $\pi(q)$ be the expected payment made by a bidder with quality $q$ in the auction. Our main result shows that the unique symmetric equilibrium bidding function in the Generalized First-Price Auction is given by the ratio of these two functions.\footnote{One cosmetic wrinkle is that when $q=0$, $\nu(q)$ and $\pi(q)$ are also equal to zero, and their ratio is of course not defined. On the other hand, the bid of a bidder with $q=0$ in GFP does not matter, so for the purposes of the statement below, we will simply set $\beta(0)=0$ for concreteness.}

\begin{theorem}
\label{th:knownq}
Bidding function $$\beta(q) = \frac{\pi(q)}{\nu(q)}$$ is the unique symmetric equilibrium of the GFP auction model of this section. 
\end{theorem}
\begin{proof}
We first show that the bidding formula in the statement of Theorem~\ref{th:knownq} does in fact constitute an equilibrium. We then show that no other symmetric equilibria exist. 

As a preliminary step, it is useful to see what the expected outcomes for bidder $i$ will be if instead of reporting his true quality score $q$ to the VCG mechanism above, he reported a different quality score $q'$. His payment to the auctioneer would be simply equal to $\pi(q')$. His expected number of clicks, however, is somewhat more subtle, since misreporting the quality score will give the bidder the same distribution of positions as that of a bidder with true quality score $q'$, and thus with the same expected visibility, but for each impression, the bidder will end up with $\frac{q}{q'}$ as many clicks. Thus, the bidder's total expected number of clicks will be equal to $\frac{q}{q'} \nu(q')$. Finally, the expected payoff of the bidder from truthful reporting is $\nu(q)-\pi(q)$, while his payoff from misreporting is $\frac{q}{q'} \nu(q') - \pi(q')$, and so by the optimality of truthful bidding in VCG, we have \begin{equation}
\label{eq:incentives}
    \nu(q)-\pi(q) \ge \frac{q}{q'} \nu(q') - \pi(q').
\end{equation}

Now, suppose all bidders follow the bidding formula $\beta(q) = \frac{\pi(q)}{\nu(q)}$. Let us first show that a bidder with a higher $q$ will be ranked higher, i.e., that the total score used for ranking bidders, $q \beta(q)$, is increasing in $q$. Take some $q' > q$. Rewrite equation~\eqref{eq:incentives} as
\begin{equation}
\label{eq:incentives2}
    \frac{\nu(q)}{q}\left(q-q\frac{\pi(q)}{\nu(q)}\right) \ge \frac{\nu(q')}{q'}\left( q - q'\frac{\pi(q')}{\nu(q')}\right).
\end{equation}
Note that $\frac{\nu(q')}{q'} > \frac{\nu(q)}{q}$, because $\nu(q')\frac{q}{q'} > \nu(q)$ (if in the VCG auction above, the bidder with quality $q$ reports $q' > q$ instead, in expectation, he will receive higher positions and more clicks). Thus, it has to be the case that $q'\frac{\pi(q')}{\nu(q')} > q\frac{\pi(q)}{\nu(q)}$, and so $q\beta(q)$ is strictly increasing in $q$.

Thus, if all bidders follow the bidding formula $\beta(q)$, those with higher quality scores will receive higher positions, and so in expectation, each bidder receives the same number of clicks and makes the same payment to the auctioneer as in the VCG auction above. 

Suppose now a bidder with quality score $q$, instead of bidding $\beta(q)$, submits a bid equal to $\frac{q'}{q}\beta(q')$ for some other quality score $q'$ (since bidders are ranked on the basis of the product of their bid and their actual quality score, it can never be optimal for a bidder to submit a bid that cannot be represented this way). Then he will receive the same distribution of positions as a bidder with quality score $q'$ who reports truthfully, and will receive $\frac{q}{q'}$ as many clicks: $\frac{q}{q'} \nu(q')$. The amount he will pay for those clicks is $\left(
\frac{q}{q'} \nu(q')
\right)\times \left(\frac{q'}{q}\beta(q')  \right) = \pi(q').$ By inequality~\eqref{eq:incentives}, such a deviation is not profitable, and thus we have shown that bidding according to the bidding function $\beta(q) = \frac{\pi(q)}{\nu(q)}$ is a symmetric equilibrium of the Generalized First-Price Auction in this setting. 

Let us now show that there are no other symmetric equilibria. To do that, it is sufficient to show that for any symmetric equilibrium bidding function $\gamma(q)$, the product $q \gamma(q)$ is strictly increasing (i.e., those with higher quality scores get higher positions in equilibrium), which would then immediately imply that the expected allocations in this alternative equilibrium are identical to those in the equilibrium derived above. By the standard revenue equivalence machinery, this would in turn imply that $\gamma(q) = \beta(q)$.

To see that $q \gamma(q)$ must be (weakly) increasing in $q$, let $\mu(q)$ be the expected number of clicks received by a bidder of quality $q$ in this equilibrium. The expected profit made by this bidder is $(1-\gamma(q))\mu(q)$. If this bidder instead submitted a bid equal to $\gamma(q')\frac{q'}{q}$ for some $q' \ne q$, he would instead make an expected profit of $\left(1-\gamma(q')\frac{q'}{q}\right)\mu(q')\frac{q}{q'}$. Since we assumed that $\gamma$ is an equilibrium bidding function, it has to be the case that 
$$ (1-\gamma(q))\mu(q) \ge \left(1-\gamma(q')\frac{q'}{q}\right)\mu(q')\frac{q}{q'}.$$
Analogously,
$$ (1-\gamma(q'))\mu(q') \ge \left(1-\gamma(q)\frac{q}{q'}\right)\mu(q)\frac{q'}{q}.$$
Adding the two inequalities and canceling the same terms on both sides, we get
$$\mu(q) + \mu(q') \ge q \frac{\mu(q')}{q'} + q' \frac{\mu(q)}q ,$$
which in turn can be rewritten as 
$$(q-q')\left(\frac{\mu(q)}q -  \frac{\mu(q')}{q'}\right) \ge 0 .$$
Finally, observe that $\frac{\mu(q)}q$ is the expected visibility of the position obtained by a bidder of quality~$q$ (because the expected number of clicks, $\mu(q)$, is equal to the product of expected visibility and the bidder's quality score $q$). Thus, we have shown that a bidder with a higher quality score $q$ must receive (weakly) higher expected visibility---which of course means that he must have a (weakly) higher total score $q \gamma(q)$. 

Finally, to show that   $q \gamma(q)$ must be strictly increasing, we use a standard argument. Suppose there is an interval of quality scores over which $q \gamma(q)$ is constant. Then a bidder with a quality score in that interval would be able to increase his expected payoff by increasing his bid by a tiny amount and ``jumping over'' a mass of competitors at an arbitrarily small cost---which of course cannot be the case. 
\end{proof}

\section{Conclusion} \label{sec:conclusion}
In this paper, we revisit the classic result on the (non-)existence of pure-strategy Nash equilibria in the Generalized First-Price Auction and show that the conclusion may be reversed when ads are ranked based on a product of stochastic quality scores and bid amounts, as is commonly done in practice---provided there is enough uncertainty in these quality scores. Moreover, the expected revenue in the pure strategy equilibrium of the Generalized First-Price Auction may substantially exceed that of the Generalized Second-Price Auction, although under some conditions the relation may also be reversed.

Our baseline model is deliberately simple and streamlined, to illustrate the essential driving forces behind our results. We also show that some of the important features of our results continue to hold in various generalizations. Of course, in practice, sponsored search auction settings can be even more complex. Nevertheless, the overall intuition of the stochastic click-through rates ``smoothing out'' payoff functions in the Generalized First-Price Auction---and thus making the existence of a pure-strategy equilibrium more likely---should continue to hold in these more general settings. 

\newpage

\bibliographystyle{chicago}
\bibliography{auctions}


\section*{Appendix}
\begin{proof}[Proof of Proposition \ref{Prop1}]

First consider \(r\le1\).  Then
\[
F(rq)=r^\theta q^\theta,
\]
and we can write 
\[
C(r)
=
A r^{(m-1)\theta}+B r^{m\theta},
\]
where
\[
A=\frac{\alpha m\theta}{m\theta+1}>0,
\qquad
B=\frac{(1-\alpha m)\theta}{(m+1)\theta+1}.
\]
Let
\[
p:=(m-1)\theta,
\qquad x:=r^\theta.
\]
Then
\[
C(r)=r^p(A+Bx).
\]
Hence
\[
C'(r)
=
r^{p-1}\bigl(pA+(p+\theta)Bx\bigr),
\]
and
\[
C(r)+rC'(r)
=
r^p\bigl((1+p)A+(1+p+\theta)Bx\bigr).
\]
Thus
\[
M(r)
=
\frac1r
\frac{pA+(p+\theta)Bx}
{(1+p)A+(1+p+\theta)Bx}.
\]
If \(B\le0\), the second factor is decreasing in $x$, so \(M\) is decreasing in $r$. 
If \(B>0\) we can differentiate $\log M(r) $ with respect to \(\log r\),
\[
\frac{d\log M(r)}{d\log r}
=
-1+
\frac{\theta^2 ABx}
{\bigl(pA+(p+\theta)Bx\bigr)
 \bigl((1+p)A+(1+p+\theta)Bx\bigr)}.
\]
 To show that $M(r)$ is decreasing it is enough to show that the second term is less than $1$. 

Note that by collecting in the denominator only terms with $AB$ and dropping some positive terms we get (after substituting $c:=p+\theta=m\theta$):
\[
\begin{aligned}
        &(pA+cBx)\left[(1+p)A+(1+c)Bx\right]  \\
        &\qquad\geq
        \left[p(1+c)+c(1+p)\right]ABx .
\end{aligned}
\]
Moreover, because \(p=(m-1)\theta\geq \theta\) and \(c=m\theta\geq \theta\), we can bound:
\[
        p(1+c)+c(1+p)=p+c+2pc>\theta^2.
\]
 Therefore,
\[
        \frac{d\log M(r)}{d\log r}<0,
\]
and \(M(r)\) is decreasing on \((0,1]\).

Now consider \(r\ge1\).  Let
\[
\mu:=\mathbb E[q]=\frac{\theta}{\theta+1}.
\]
For \(r\ge1\), bidder \(i\) surely gets the top slot whenever \(q\ge1/r\).  A change of variables gives
\[
C(r)=\mu-K r^{-(\theta+1)},
\]
where
\[
K:=\mu-C(1)\ge0.
\]
Therefore
\[
C'(r)=(\theta+1)K r^{-(\theta+2)}.
\]
Thus
\[
M(r)
=
\frac{(\theta+1)K r^{-(\theta+2)}}
{\mu+\theta K r^{-(\theta+1)}}.
\]
If \(K=0\), then \(M(r)=0\).  If \(K>0\), then
\[
\frac{d\log M(r)}{d\log r}
=
-(\theta+2)
+
\frac{(\theta+1)\theta K r^{-(\theta+1)}}
{\mu+\theta K r^{-(\theta+1)}}.
\]
The second term is strictly smaller than \(\theta+1\), so $M(r)$ is decreasing for $r>1$ too. 

The two formulas agree at \(r=1\), since \(C\) is continuously differentiable.  Hence \(M(r)\) is weakly decreasing on the entire feasible domain.  Applying Lemma \ref{lemma1SC} establishes that $b^*$ is the symmetric pure-strategy equilibrium.
\end{proof}
\medskip

\begin{proof}[Proof of Proposition \ref{Prop_large_n_expon}]

To establish existence in the exponential case, we start with a technical lemma:
\begin{lemma}
Let \(m\geq 2\), \(\alpha\in(0,1)\), and
\[
        G(x)
        =
        (1-e^{-x})^m+\alpha m e^{-x}(1-e^{-x})^{m-1}.
\]
Then \((r,q)\mapsto G(rq)\) is log-concave on \(\mathbb R_{++}^2\).
\end{lemma}

\begin{proof}
Let
\[
        g(x)=\log G(x),
        \qquad
        y=e^{-x},
        \qquad
        c=\alpha m-1.
\]
Then \(c\in(-1,m-1)\) and
\[
        G(x)=(1-y)^{m-1}(1+cy).
\]
Let
\[
        \Delta:=(1-y)(1+cy)>0.
\]
A direct differentiation gives
\[
        g'(x)=\frac{yQ}{\Delta},
        \qquad
        g''(x)=-\frac{yP}{\Delta^2},
\]
where
\[
        Q=(m-1)+c(my-1),
\]
and
\[
        P=(m-1)(1+c^2y^2)+c(2my-y^2-1).
\]

First, we claim that \(Q>0\). To see this, note that if \(c\geq 0\), then
\[
        Q\geq m-1-c=m(1-\alpha)>0.
\]
On the other hand, if \(c<0\), then
\[
        Q\geq (m-1)(1+c)=\alpha m(m-1)>0.
\]

Next, define
\[
        S:=2P-Q(1+cy).
\]
We claim that \(S\geq 0\). If \(c\geq 0\), then
\[
        S
        =
        (m-1-c)(1+cy)
        +
        cy\left[c((m-2)y+2)+m+2-2y\right]\geq 0.
\]
If \(c<0\), write \(c=-d\), where \(d\in(0,1)\). Then
\[
        S
        =
        (1-d)(m-1)+dm(1-y)^2
        -
        d(1-d)y((m-2)y+1).
\]
Since \(m\geq 2\) and \(0<y<1\),
\[
        y((m-2)y+1)\leq m-1.
\]
Therefore
\[
        S
        \geq
        (1-d)^2(m-1)+dm(1-y)^2
        \geq 0.
\]

Because \(Q>0\), \(1+cy>0\), and \(S=2P-Q(1+cy)\geq 0\), we have \(P>0\). Hence
\[
        g''(x)=-\frac{yP}{\Delta^2}\leq 0.
\]
Moreover, since \(x=-\log y\) and \(-\log y\geq 1-y\),
\[
\begin{aligned}
        2xP
        &=
        2(-\log y)P  \\
        &\geq
        2(1-y)P       \\
        &\geq
        Q(1-y)(1+cy)
        =
        Q\Delta.
\end{aligned}
\]
Thus
\[
        g'(x)+2xg''(x)
        =
        \frac{y}{\Delta^2}\left(Q\Delta-2xP\right)
        \leq 0.
\]

Now consider
\[
        \Phi(r,q):=\log G(rq)=g(rq).
\]
Let \(x=rq\). The Hessian of \(\Phi\) is
\[
        \boldsymbol{H}(\Phi(r,q))
        =
        \begin{pmatrix}
        q^2g''(x) & g'(x)+xg''(x) \\
        g'(x)+xg''(x) & r^2g''(x)
        \end{pmatrix}.
\]
The diagonal entries are nonpositive because \(g''(x)\leq 0\). Its determinant is
\[
\begin{aligned}
        q^2r^2g''(x)^2-\left(g'(x)+xg''(x)\right)^2
        &=
        x^2g''(x)^2-\left(g'(x)+xg''(x)\right)^2  \\
        &=
        -g'(x)\left(g'(x)+2xg''(x)\right)
        \geq 0,
\end{aligned}
\]
because \(g'(x)>0\) and \(g'(x)+2xg''(x)\leq 0\). Therefore $\boldsymbol{H}(\Phi)$ is negative
semidefinite, so \(\Phi\) is concave. Hence \((r,q)\mapsto G(rq)\) is log-concave.
\end{proof}

With this lemma, we proceed by showing that $C(r)$ is log-concave in this case and hence by Lemma \ref{Lemma2}, $b^*$ is an equilibrium.  Normalize \(\lambda=1\), since it does not affect log-concavity of $C(r)$. 

Conditional on own quality \(q\), an opponent is below bidder \(i\)'s score with probability
\[
        1-e^{-rq}.
\]
Therefore bidder \(i\)'s expected number of clicks from relative bid \(r\) is
\[
        C(r)
        =
        \int_0^\infty
        qe^{-q}
        \left[
        (1-e^{-rq})^m+\alpha m e^{-rq}(1-e^{-rq})^{m-1}
        \right]dq.
\]
Define
\[
        G(x)
        =
        (1-e^{-x})^m+\alpha m e^{-x}(1-e^{-x})^{m-1}.
\]
Then
\[
        C(r)=\int_0^\infty qe^{-q}G(rq)\,dq.
\]

By the technical lemma above, \((r,q)\mapsto G(rq)\) is log-concave. Also,
\[
        q\mapsto qe^{-q}
\]
is log-concave on \(\mathbb R_{++}\), since $\log(qe^{-q})=\log q-q$ is concave. Hence the full integrand
\[
        (r,q)\mapsto qe^{-q}G(rq)
\]
is jointly log-concave in \((r,q)\).

By Prékopa's theorem, integrating out \(q\) preserves log-concavity. Therefore
\[
        C(r)=\int_0^\infty qe^{-q}G(rq)\,dq
\]
is log-concave in \(r\), and Lemma \ref{Lemma2} establishes the claim.
\end{proof}

\medskip

\begin{proof}[Proof of Proposition \ref{Prop_revenue_exponential}]
For exponential distributions, the expected order statistics are
\[
\mathbb E[q_{(j)}]=\frac{H_n-H_{j-1}}{\lambda},
\]
where
\[
H_n=\sum_{s=1}^n \frac1s,
\qquad
H_n^{(2)}=\sum_{s=1}^n \frac1{s^2},
\qquad
H_0=H_0^{(2)}=0.
\]

\paragraph{GFP revenue.}
We can compute directly,
\[
C(1)=\frac{1}{n\lambda}\left[H_n+\alpha(H_n-1)\right].
\]
Define
\[
A_{n,\alpha}:=H_n+\alpha(H_n-1)
=(1+\alpha)H_n-\alpha.
\]
Then
\[
C(1)=\frac{A_{n,\alpha}}{n\lambda}.
\]

Next,
\[
C'(1)
=
\frac{1}{n\lambda}
\Bigg\{
(1-\alpha)
\left[
(H_n-1)^2+H_n^{(2)}-1
\right]
+
2\alpha
\left[
\left(H_n-\frac32\right)^2+H_n^{(2)}-\frac54
\right]
\Bigg\}.
\]
Define
\[
B_{n,\alpha}
:=
(1-\alpha)
\left[
(H_n-1)^2+H_n^{(2)}-1
\right]
+
2\alpha
\left[
\left(H_n-\frac32\right)^2+H_n^{(2)}-\frac54
\right].
\]
Equivalently,
\[
B_{n,\alpha}
=
(1+\alpha)\left(H_n^2+H_n^{(2)}\right)
-
2(1+2\alpha)H_n
+
2\alpha.
\]
Then
\[
C'(1)=\frac{B_{n,\alpha}}{n\lambda},
\]
and therefore
\[
b^*
=
\frac{B_{n,\alpha}}{A_{n,\alpha}+B_{n,\alpha}}.
\]

In the symmetric GFP equilibrium, all bidders bid the same \(b^*\), so the allocation is efficient.  Total expected clicks equal
\[
nC(1)=\frac{A_{n,\alpha}}{\lambda}.
\]
Thus GFP revenue is
\[
REV_{GFP}
=
b^* nC(1).
\]
Substituting,
\[
REV_{GFP}
=
\frac{1}{\lambda}
\frac{A_{n,\alpha}B_{n,\alpha}}
{A_{n,\alpha}+B_{n,\alpha}}.
\]

\paragraph{VCG revenue.}
As before, the expected revenue in the VCG auction is
\[
REV_{VCG}
=
(1-\alpha)\mathbb E[q_{(2)}]
+
2\alpha \mathbb E[q_{(3)}].
\]
Using the exponential order-statistic formula,
\[
\mathbb E[q_{(2)}]=\frac{H_n-1}{\lambda},
\qquad
\mathbb E[q_{(3)}]=\frac{H_n-\frac32}{\lambda}.
\]
Hence
\[
REV_{VCG}
=
\frac{1}{\lambda}
\left[
(1-\alpha)(H_n-1)
+
2\alpha\left(H_n-\frac32\right)
\right]
=
\frac{1}{\lambda}
\left[
(1+\alpha)H_n-(1+2\alpha)
\right].
\]

\paragraph{Revenue comparison.}
Using the expressions above,
\[
REV_{GFP}-REV_{VCG}
=
\frac{1}{\lambda}
\left[
\frac{A_{n,\alpha}B_{n,\alpha}}
{A_{n,\alpha}+B_{n,\alpha}}
-
\left((1+\alpha)H_n-(1+2\alpha)\right)
\right]
=
\frac{
\alpha-(1+\alpha)^2\left(H_n-H_n^{(2)}\right)
}{
\lambda\left(A_{n,\alpha}+B_{n,\alpha}\right)
}.
\]
The denominator is strictly positive because
\[
A_{n,\alpha}=n\lambda C(1)>0,
\qquad
B_{n,\alpha}=n\lambda C'(1)>0.
\]
Thus the sign of \(REV_{GFP}-REV_{VCG}\) is the same as the sign of
\[
\alpha-(1+\alpha)^2\left(H_n-H_n^{(2)}\right).
\]

Now observe that
\[
H_n-H_n^{(2)}
=
\sum_{s=2}^n\left(\frac1s-\frac1{s^2}\right)
=
\sum_{s=2}^n \frac{s-1}{s^2}.
\]
Since \(n>2\),
\[
H_n-H_n^{(2)}
>
\frac14.
\]
Also, for every \(\alpha\in[0,1]\),
\[
\frac{\alpha}{(1+\alpha)^2}\le \frac14.
\]
That establishes the claim.
\end{proof}
\medskip

\begin{proof}[Proof of Proposition \ref{Prop_more_than_2_slots}]
By Lemma \ref{Lemma2}, it is sufficient to show that $1/C(r)$ is convex. Let
\[
x=q^\theta,\qquad p=\frac1\theta.
\]
Then $q=x^p$ and $f(q)dq=dx$. We can rewrite (\ref{n_slots_Cr}) as:
\[
C(r)
=
\int_0^1
x^p
\left[\alpha+(1-\alpha)\min\{r^\theta x,1\}\right]^m dx.
\]

First consider $r\le 1$. Set $z=r^\theta$. Then
\[
C(r)=\Phi(z),
\]
where
\[
\Phi(z)
=
\int_0^1
x^p\left[\alpha+(1-\alpha)zx\right]^m dx.
\]
Equivalently,
\[
\Phi(z)
=
\sum_{\ell=0}^m
\binom{m}{\ell}
\alpha^{m-\ell}(1-\alpha)^\ell
\frac{z^\ell}{p+\ell+1}.
\]
A direct differentiation gives
\[
\left(\frac1{\Phi(z)}\right)''
=
\frac{2(\Phi'(z))^2-\Phi(z)\Phi''(z)}{\Phi(z)^3}
\ge 0,
\]
(where the inequality is established in Lemma \ref{Lemma4} below).
Thus $1/\Phi(z)$ is convex and decreasing in $z$. Since $0<\theta\le 1$, and $r^\theta$ is concave, \[
\frac1{C(r)}
=
\frac1{\Phi(r^\theta)}
\]
is convex on $[0,1]$.

Now consider $r\ge 1$. Let
\[
\mu=\mathbb E[q]=\frac{\theta}{\theta+1},
\qquad
C_1=C(1).
\]
For $r\ge1$, the expression truncates at $q=1/r$, and a change of variables gives
\[
C(r)=\mu-(\mu-C_1)r^{-(\theta+1)}.
\]
Therefore $C'(r)>0$ and $C''(r)<0$ for $r>1$, so
\[
\left(\frac1{C(r)}\right)''
=
\frac{2(C'(r))^2-C(r)C''(r)}{C(r)^3}
>0.
\]
Hence $1/C(r)$ is convex also for $r \geq 1$. 

The two formulas agree at $r=1$, and $C$ is continuously differentiable there, so $1/C(r)$ is convex on the whole feasible domain. By Lemma \ref{Lemma2}, $b^*$ is a pure strategy equilibrium.   
\end{proof}

\begin{lemma}[Convexity of $1/\Phi$]
\label{Lemma4}
Let $m\ge1$, $p\ge0$, and $\alpha\in(0,1)$, and define
\[
\Phi(z)
=\int_0^1 x^p\bigl[\alpha+(1-\alpha)zx\bigr]^m\dd x,
\qquad z\ge0.
\]
Then $1/\Phi(z)$ is strictly convex on $[0,\infty)$.
\end{lemma}

\begin{proof}
Fix $z\ge 0$, let $X$ have density
$x^p\bigl(\alpha+(1-\alpha)zx\bigr)^m/\Phi(z)$ on $[0,1]$, and set
$T:=X/\bigl(\alpha+(1-\alpha)zX\bigr)$. Differentiating under the
integral sign, $\Phi'/\Phi=m(1-\alpha)\,\mathbb E[T]$ and
$\Phi''/\Phi=m(m-1)(1-\alpha)^2\,\mathbb E[T^2]$, so
\[
2(\Phi')^2-\Phi\Phi''
=m(1-\alpha)^2\Phi^2
\Bigl(2m\bigl(\mathbb E[T]\bigr)^2-(m-1)\,\mathbb E[T^2]\Bigr),
\]
and since $\mathbb E[T]>0$, it suffices to show
$\mathbb E[T^2]\le 2\bigl(\mathbb E[T]\bigr)^2$.

The substitution $t=x/\bigl(\alpha+(1-\alpha)zx\bigr)$ shows that $T$
has density
$f_T(t)\propto t^p\bigl(1-(1-\alpha)zt\bigr)^{-(p+m+2)}$ on $[0,L]$,
where $L:=1/\bigl(\alpha+(1-\alpha)z\bigr)$. As
$1-(1-\alpha)zL=\alpha/\bigl(\alpha+(1-\alpha)z\bigr)>0$, this density
is bounded, and it is nondecreasing. Hence
$F_T(t)/t=\frac1t\int_0^t f_T$ is nondecreasing in $t$ (its derivative
has the sign of $tf_T(t)-\int_0^t f_T\ge 0$), so $F_T(t)\le t/L$ on
$[0,L]$, and
\[
\mathbb E[T]=\int_0^L\bigl(1-F_T(t)\bigr)\,dt
\ \ge\ \int_0^L\Bigl(1-\frac tL\Bigr)dt=\frac L2 .
\]
Since $T\le L$ almost surely,
\[
\mathbb E[T^2]\ \le\ L\,\mathbb E[T]\ \le\ 2\bigl(\mathbb E[T]\bigr)^2 ,
\]
so $2(\Phi')^2-\Phi\Phi''\ge
2m(1-\alpha)^2\Phi^2\bigl(\mathbb E[T]\bigr)^2>0$, and
$\bigl(1/\Phi\bigr)''=\bigl(2(\Phi')^2-\Phi\Phi''\bigr)/\Phi^3>0$.

\end{proof}

\end{document}